\documentclass[12pt]{article}

\usepackage{srcltx}
\usepackage[dvips,letterpaper]{geometry}
\usepackage{graphicx,psfrag,epsf}
\usepackage{times}
\usepackage{fullpage}
\usepackage{setspace}
\usepackage{amsmath,amssymb,amsfonts}
\usepackage{bm}
\usepackage{graphicx}
\usepackage{epsfig}
\usepackage{caption}
\usepackage{subfigure}
\usepackage{url}
\usepackage{verbatim}
\usepackage{natbib}
\usepackage{color}
\usepackage{booktabs}
\usepackage{graphicx,psfrag,epsf}
\usepackage{lscape}
\usepackage{adjustbox}
\usepackage{rotating}
\usepackage{pdflscape}

\input cyracc.def

\newtheorem{Proposition}{Proposition}[section]
\newtheorem{Remark}{Remark}[section]

\newtheorem{proof}{Proof}
\title{The shared weighted Lindley frailty model for cluster failure time data}

\author{Diego I. Gallardo$^1$ \,and\, Marcelo Bourguignon$^2$\thanks{Corresponding
author: Marcelo Bourguignon. Department of Statistics, Universidade Federal do Rio Grande do Norte, Natal, RN, Brazil. Email: m.p.bourguignon@gmail.com.}   \\
{\footnotesize $^1$Department of Mathematics, Facultad de Ingenier\'ia, Universidad de Atacama, Copiap\'o, Chile}\\[-0.15cm]
{\footnotesize $^2$Department of Statistics, Universidade Federal do Rio Grande do Norte, Natal, RN, Brazil}\\[-0.15cm]
}

\date{}

\begin{document}
%\linenumbers

\maketitle

\vspace{-0.7cm}
\begin{abstract}
The primary goal of this paper is to introduce a novel frailty model based on the weighted Lindley (WL) distribution for modeling clustered survival data. We study the statistical properties of the proposed model. In particular, the amount of unobserved heterogeneity is directly parameterized on the variance of the frailty distribution such as gamma and inverse Gaussian frailty models. Parametric and semiparametric versions of the WL frailty model are studied. %For the semiparametric version, we use the piecewise constant baseline hazard function within the proportional hazards model in a frailty framework.
A simple expectation-maximization (EM) algorithm is proposed for parameter estimation. Simulation studies are conducted to evaluate its finite sample
performance. Finally, we apply the proposed model to a real data set to analyze times after surgery
in patients diagnosed with colorectal cancer and compare our results with classical frailty models carried out in this application, which shows the superiority of the proposed model. We implement an R package that includes estimation for fitting the proposed model based on the EM-algorithm.
\vspace{-0.3cm}

%\vspace{-0.3cm}
%\paragraph{Mathematics subject classification} Primary 62F10; Secondary 62F12.
\end{abstract}

\section{Introduction}
\noindent

In survival analysis, when unobserved sources of heterogeneity are present in the data,
the usual statistical approach known as Cox proportional hazards model \citep{COX} is not appropriate.
In this case, frailty models \citep{Vaupel1979} can be used for modeling unobserved
heterogeneity among subjects or groups, which is usually due to random effects and/or
omitted covariates in the study. Frailty models are characterized by the inclusion of a latent
random effect containing data that cannot be measured or have not been observed.

Various authors discussed frailty models. The gamma \citep{Vaupel1979, Congdon1995} and
inverse Gaussian \cite[IG; ][]{Hougaard84, Manton86} distributions
are the most commonly used frailty distributions because of their mathematical convenience. However, the gamma distribution
has just a monotone hazard rate function. On the other hand, the IG has an upside-down bathtub hazard rate function.
\cite{Hougaard86a} used the positive stable distribution for the frailty, but its density function is intractable.
Other possibilities are the log-normal \citep{Flinn1982} and Birnbaum-Saunders \cite{Leao17} distributions.
The log-normal frailty model does not have a known Laplace transform, thus the likelihood function becomes intractable.
The Birnbaum-Saunders frailty model has a mathematically tractable Laplace transform, but its variance is limited
\citep{mota21}. Excellent reviews of frailty models are given by \cite{Wienke2011}, and \cite{Hanagal}.

Shared or clustered failure time data \citep{Hougaard86b} are very common in survival analysis, and
the idea of frailty models framework can be extended to the shared case.
In this context, many distributions have been considered in the literature.
Generalized gamma has been introduced as frailty distribution by \cite{balagg}.
\cite{Balakrishnan18} proposed the semi-parametric likelihood inference for the shared Birnbaum-Saunders frailty model.
Recently, \cite{Barreto2019} and \cite{Piancastelli2021} proposed the generalized exponential (EE) and the
generalized inverse Gaussian (GIG) frailty models, respectively, for clustered survival data.
However, in last two models do not fixed the mean of the frailty distribution at 1 as usually is used in this context. In these cases, we cannot compare frailty terms with the usual models, such as gamma and IG
frailty models. Furthermore, in both models, the derivatives of the Laplace transform do not have a closed form, which difficult its application for data with clusters with a large number of observations. In addition, as the frailty terms are not centered at the same point, comparing the variances does not make sense either.

\cite{GHITANY20111190} introduced the two-parameter weighted Lindley (shortly WL) distribution in order to model failure time such as the Birnbaum-Saunders, gamma, inverse Gaussian, lognormal and Weibull distributions.
The probability density function of the WL model adds an extra shape that can be useful for modeling bimodal data, which cannot be modeled using the gamma, IG or Weibull distributions. Furthermore, the WL model has a bathtub or an increasing hazard rate function depending on the values of its parameters, which cannot be modeled using the gamma, IG or Weibull distributions. These characteristics of the WL distribution motivate us to use it as frailty distribution.
In this context, the application of the WL distribution in frailty models in a univariate context was considered in \cite{mota21}. Recently, \cite{Tyagi21} studied the bivariate case.

In this paper, we use the WL model as
the frailty distribution for clustered survival data. Both parametric and semiparametric versions of the WL frailty model are studied within the proportional hazards model to come up with a flexible frailty model. It has a closed-form for the conditional likelihood function, given
the observed data, so that an EM algorithm can be applied effectively to obtain
the maximum likelihood (ML) estimates.
Hereunder, we list some of the main contributions and advantages of the proposed frailty model.
\begin{enumerate}
\item Mathematical simplicity of our model: the unconditional density, survival, and hazard functions related to the WL frailty model have closed forms and are very simple. Furthermore, the conditional distribution of frailties among the survivors and the frailty of individuals dying at time $t$ can be explicitly determined (WL distributed). Finally, the derivatives of the Laplace transform for the WL distribution has closed-form in contrast with EE, GIG, and other models;

\item Properties simplicity: the probability and distribution functions of the WL model have a simple form in contrast with other frailty models which have associated probability function involving special functions (beta or modified Bessel functions);

\item Flexibility: the WL model is suitable for modeling right skewed positive data with bathtub shaped hazard rate function, which cannot be modeled using the gamma or Weibull distributions. Furthermore, WL distribution is suitable for modeling bimodal data which cannot be modeled using the gamma, IG or Weibull distributions;

\item Special case: the Lindley frailty model is a special case of the WL frailty model;

\item Model estimation: we found the ML estimators through an expectation-maximization (EM) algorithm. In particular, we provide a simple EM-algorithm, since all conditional expectations involved in the E-step are obtained in explicit form;

\item Applications: the Monte Carlo simulations and empirical application show the good performance of the proposed frailty model (see Section \ref{aplic}).
\end{enumerate}

This paper is organized as follows. In Section 2,  we present a brief summary of the WL frailty models and
propose the shared WL frailty models.
In addition, we provide new general properties of the WL frailty model.
Estimation of the parameters by maximum likelihood (ML) estimation via EM-algorithm and a semiparametric
approach is investigated in Section 3.
In Section 4, some numerical results of the estimators are presented with a discussion of the results.
The proposed model is illustrated with the time after surgery in patients diagnosed with colorectal cancer in Section 5.
It is shown that the proposed model has a better performance than those based on gamma and inverse Gaussian frailty models.
Discussions and some concluding remarks are shown in Section 6.
The computational functions to fit the WL frailty model were implemented in
the R programming language \citep{R2022} and were compiled into an initial
version of an R package called \texttt{extrafrail}, available at
https://CRAN.R-project.org/package=extrafrail.

\section{The state of the art for the WL frailty models}
\noindent

In this Section, we present briefly the WL distribution and its use in a frailty models context.

%Let $T$ be a survival time with an absolutely continuous distribution whose hazard function, conditional on a latent variable $Z$, satisfies $h(t|Z) = Z\,h_0(t)$ for $t > 0$, where $h_0(t)$ is a baseline hazard function and $Z$ is the frailty of the individual. In this paper, we assume that the frailty $Z$ has a WL distribution.

\subsection{The WL distribution}\label{sec2.1}
\noindent

The WL distribution was studied in \cite{GHITANY20111190}.
A random variable $Z$ follows a WL distribution with parameters $\alpha$ and $\phi$, denoted by $Z \sim \textrm{WL}(\alpha, \phi)$, if its probability density function (pdf) is given by
%if its cumulative distribution function is given by
%\begin{equation*}\label{cdf}
%F(z; \alpha, \phi) =1 - \frac{(\alpha + \phi)\Gamma(\phi, \alpha\,z) + %(\alpha\,z)^\phi\textrm{e}^{-\alpha\,z}}{(\alpha+\phi)\Gamma(\phi)}, \quad {z} > 0,
%\end{equation*}
%where $\alpha > 0$ is a scale parameter, $\phi > 0$ is the shape parameter, $\Gamma(a, b) = \int^{\infty}_{b}x^{a-1}\textrm{e}^{-x}\textrm{d} x$ is the upper incomplete gamma function and $\Gamma(a) = \int^{\infty}_{0}x^{a-1}\textrm{e}^{-x}\textrm{d} x$ is the gamma function. The corresponding probability density function is
\begin{equation}\label{inv:01}
f(z; \alpha, \phi) = \frac{\alpha^{\phi+1}}{(\alpha + \phi)\Gamma(\phi)}z^{\phi-1}(1+z)\textrm{e}^{-\alpha\,z}, \quad {z} > 0,
\end{equation}
where $\alpha > 0$ is a scale parameter and $\phi > 0$ is the shape parameter. For $\phi = 1$ the WL distribution reduces to Lindley distribution. For $\phi < 1$, the WL distribution is suitable for modeling right skewed positive data with bathtub shaped hazard rate function, which cannot be modeled using, for instance, the gamma or Weibull models. Furthermore, the WL distribution is also suitable for modeling bimodal data, which cannot be modeled using the gamma or Weibull distributions.
%
%The $k$th moment about zero of $Z$ is given by
%\begin{equation*}\label{inv:04}
%\mathbb{E}(Z^k) = \frac{(\alpha + \phi + k)\Gamma(\phi+k)}{\alpha^k(\alpha + \phi)\Gamma(\phi)}, \quad k = 1, 2, %\ldots
%\end{equation*}

In particular, the mean and variance associated with \eqref{inv:01}, are respectively given by
\begin{equation*}\label{inv:05}
\mathbb{E}(Z) = \frac{\phi(\alpha + \phi + 1)}{\alpha(\alpha + \phi)} \qquad \mbox{and} \qquad \textrm{Var}(Z) = \frac{(\phi +1 )(\alpha + \phi)^2 - \alpha^2}{\alpha^2(\alpha + \phi)^2}.
\end{equation*}
Additionally, a useful result for our development is
\begin{equation*}
\mathbb{E}(\log Z) = -\frac{\alpha}{\phi(\alpha+\phi)}+\psi(\phi+1)-\log(\alpha),
\end{equation*}
where $\psi(\cdot)$ denotes the digamma function. %The proposed frailty model is closely related to the gamma frailty model (Nielsen et al., 1992).
The WL distribution can be viewed as a mixture of two gamma distributions with known weights \citep{GHITANY20111190} as follows
\begin{eqnarray*}\label{mixture}
f(z; \alpha, \phi) &=& \omega \, \frac{\alpha^{\phi}}{\Gamma(\phi)}z^{\phi-1}\textrm{e}^{-\alpha\,z} + (1-\omega)\,
\frac{\alpha^{\phi+1}}{\Gamma(\phi+1)}z^{\phi}\textrm{e}^{-\alpha\,z}
= \omega \,f_{Z_1}(z)+ (1-\omega)\,f_{Z_2}(z),
\end{eqnarray*}
where $\omega = \alpha/(\alpha + \phi)$ and $f_{Z_1}(z)$ and $f_{Z_2}(z)$
are the pdf of the $Z_1 \sim \textrm{Ga}(\phi, \alpha)$ and $Z_2 \sim \textrm{Ga}(\phi + 1, \alpha)$, respectively. This mixture of gamma distributions has a certain advantage over competitors since it does not require a subjective approach involving guessing the mixing weights (know weights), which is a useful property of the proposed model.
The application of the WL distribution in a frailty models in a univariate context was considered in \cite{mota21} having $\alpha = \sqrt{\phi(\phi+1)}$. With this restriction, we have that $\mathbf{E}(Z) = 1$ and the variance of $Z$ is given by $\theta := 2(\phi+\sqrt{\phi(\phi+1)})^{-1}$ (i.e., $\phi=4/(\theta(\theta+4))$). For this reason, henceforth we consider the parametrization in terms of $\theta$. From here on, we use the notation $Z \sim \textrm{WL}(\theta)$ to indicate that $Z$ is a random variable following a reparameterized WL distribution.
Consequently, the pdf and Laplace transform for this particular WL model are, respectively,
\begin{align*}\label{density.WL}
f(z; \theta) &= \frac{\theta}{2 \Gamma(b_{\theta})} a_{\theta}^{-b_{\theta}-1}z^{b_{\theta}-1}(1+z)\exp\left\{-\frac{z}{a_{\theta}}\right\}, \quad {z}>0,\\
\mathcal{L}_Z(s) &=  \mathbb{E}(e^{-s Z})=\left(1 + a_{\theta} s\right)^{-b_{\theta}-1}\left(1 + \frac{ \theta s}{2}\right), \quad s\in \mathbb{R},
\end{align*}
where $f(z; \theta):=f(z; \sqrt{4/(\theta(\theta+4)[4/(\theta(\theta+4))+1]}, 4/(\theta(\theta+4)))$, $a_{\theta}=\frac{ \theta (\theta+4)}{2(\theta+2)}$ and $b_{\theta}=\frac{4}{\theta(\theta+4)}$.
Thus, \begin{eqnarray*}\label{mixture}
f(z; \theta) &=& \omega \, \frac{a_\theta^{-b_\theta}}{\Gamma(b_\theta)}z^{b_\theta-1}\textrm{e}^{-z/a_\theta} + (1-\omega)\,
\frac{a_\theta^{-b_\theta-1}}{\Gamma(b_\theta+1)}z^{b_\theta}\textrm{e}^{-z/a_\theta}
= \omega \,f_{Z_1}(z)+ (1-\omega)\,f_{Z_2}(z),
\end{eqnarray*}
where $\omega = a_\theta^{-1}/(a_\theta^{-1} + b_\theta)$ and $f_{Z_1}(z)$ and $f_{Z_2}(z)$
are the pdf of the $Z_1 \sim \textrm{Ga}(b_\theta, a_\theta^{-1})$ and $Z_2 \sim \textrm{Ga}(b_\theta + 1, a_\theta^{-1})$, respectively.

To finalize this subsection, in the following Proposition, we present the derivatives of the Laplace transform for the WL$(\theta)$ model. This result is very useful to our future development.
\begin{Proposition}\label{d.Lz}
For the WL$(\theta)$ model and for $d \in \mathbb{N}$ the $d$-th derivative in relation to $s$ of $\mathcal{L}_Z(s)$, say $\mathcal{L}^{(d)}_Z(s)$, is given by
\begin{align*}
\mathcal{L}^{(d)}_Z(s)&=(-1)^d \pi_d(b_{\theta}) a_{\theta}^{d-1} \left(1 + a_{\theta} s\right)^{-b_{\theta}-d-1}\left(1+\frac{\theta(s+d-1)}{\theta+2}\right),
\end{align*}
where $\pi_1(b_{\theta})=1$ and $\pi_d(b_{\theta})=\prod_{i=1}^{d-1}\left\{b_{\theta}+i\right\}$, for $d>1$.
\end{Proposition}

\begin{proof}
The proof is simple using induction on $d$.
\end{proof}

\subsection{WL frailty models in the literature}
\noindent

In this subsection we present a brief summary on the WL frailty models in the literature in order to clarify our contribution.

\subsubsection{Univariate WL frailty model}
\noindent

Let $Z_i > 0$, $i=1,\ldots,n$, be the latent random variable representing the frailty term associated to the $i$-th individual. In a multiplicative hazards framework, given $Z_i = z_i$ and a vector of $p$ covariates (without intercept term), say $\mathbf{x}_i=(x_{i1},\ldots,x_{ip})$, the conditional hazard function for the $i$-th individual is given by
\begin{equation}
\lambda(t\mid z_i, \mathbf{x}_i)=\lambda_0(t) z_i \exp(\mathbf{x}_i^\top {\bm \beta}), \quad i=1,\ldots,n, \label{cond.lambda}
\end{equation}
where ${\bm \beta}=(\beta_1,\ldots,\beta_p)$ are the regression parameters, respectively, and the distribution of $Z$ corresponds to a nonnegative random variable. %Many distributions have been considered in the literature: gamma \citep{Vaupel1979, Congdon1995}, log-normal \citep[LN; ][]{Flinn1982}, inverse gaussian \cite[IG; ][]{Hougaard84, Manton86}, the positive stable \cite[PS, page 105, ][]{Wienke2011}, Birnbaum-Saunders \cite[BS; ][]{Leao17, Balakrishnan18}, generalized exponential \cite[GE; ][]{Barreto2019}, generalized inverse gaussian \cite[GIG; ][]{Piancastelli2021} and recently, the WL model \citep{mota21}.
The conditional survival function related to (\ref{cond.lambda}) is given by

\begin{equation*}
S(t\mid z_i, \mathbf{x}_i)=\exp\left(-z_i  \exp\left(\mathbf{x}_i^\top {\bm \beta}\right) \Lambda_0(t)\right), \quad i=1, \ldots, n, \label{cond.S}
\end{equation*}
and the marginal survival function (obtained integrating eq. (\ref{cond.S}) in relation to the density function assumed for $Z$) is given by
\begin{equation}
S(t\mid \mathbf{x}_i)=\mathcal{L}_Z\left(\exp\left(\mathbf{x}_i^\top {\bm \beta}\right) \Lambda_0(t)\right), \quad i=1,\ldots,n, \label{cond.S}
\end{equation}
and the corresponding marginal pdf is
\begin{equation*}
f(t\mid \mathbf{x}_i) =  -\exp\left(\mathbf{x}_i^\top {\bm \beta}\right) \lambda_0(t) \mathcal{L}_Z^{(1)}\left(\exp\left(\mathbf{x}_i^\top {\bm \beta}\right) \Lambda_0(t)\right), \quad i=1,\ldots,n.
\end{equation*}
Note that all of the mentioned distributions have a closed form to the Laplace transform and hence their use in this specific context. Particularly, for $Z_i\sim \mbox{WL}(\theta)$ model, such marginal functions assume the forms
\begin{align}
f(t\mid \mathbf{x}_i) &=  \lambda_0(t) \exp\left(\mathbf{x}_i^\top {\bm \beta}\right)  \left(1 + a_{\theta} \exp\left(\mathbf{x}_i^\top {\bm \beta}\right) \Lambda_0(t)\right)^{-b_{\theta}-2}\left(1 + \frac{ \theta \exp\left(\mathbf{x}_i^\top {\bm \beta}\right) \Lambda_0(t)}{(\theta+2)}\right) \nonumber, \quad \mbox{and}\\
S(t\mid \mathbf{x}_i) &=  \left(1 + a_{\theta} \exp\left(\mathbf{x}_i^\top {\bm \beta}\right) \Lambda_0(t)\right)^{-b_{\theta}-1}\left(1 + \frac{ \theta \exp\left(\mathbf{x}_i^\top {\bm \beta}\right) \Lambda_0(t)}{2}\right), \quad \mbox{for } t,\theta>0. \label{marg.surv.WL}
\end{align}

\noindent For this particular model, we also present the following new additional results.

%The unconditional survival function can be derived by the Laplace transform
%\begin{equation*}
%S(t) = L_Z(H_0(t)) = \frac{1}{2}(1 + \theta H_0(t))^{-1/\theta}\left[1 + \left(1 + \theta %H_0(t)\right)^{-1}\right],
%\end{equation*}
%where $H_0(t) = \int_{0}^{t}h_0(u)\textrm{d}u$ is the cumulative hazard function.
%This implies the unconditional probability density function
%\begin{eqnarray*}
%f(t) &=& -h_0(t)L^{'}_Z(H_0(t)) = \frac{h_0(t)}{2}(1 + \theta H_0(t))^{-1/\theta-1}\left[1 + (\theta + 1)\left(1 %+ \theta H_0(t)\right)^{-1}\right],
%\end{eqnarray*}
%where $L^{'}_Z(s) = \textrm{d}\, L_Z(s)/\textrm{d}\,s$.

\begin{Proposition}\label{prop1}
The density of the frailty distribution among the survivors (indicated by the condition $T > t$)
can be written in the form
\begin{eqnarray*}
f(z|T> t) &=&
\frac{A_\theta^{-b_\theta-1}}{(A_\theta^{-1} + b_\theta)\Gamma(b_\theta)}z^{b_\theta - 1}(1+z)\exp(-z/A_\theta), \quad {z}>0,
\end{eqnarray*}
which is the density of a WL$(b_\theta, A_\theta^{-1})$, where $A_\theta = a_\theta/(1+a_\theta\,\Lambda_0(t))$.
\end{Proposition}

\begin{Proposition}\label{prop2}
The density of the frailty given a failure at time $t$, that is the conditional
distribution of $Z|T = t$, is given by
\begin{eqnarray*}
f(z|T = t) &=& \frac{A_\theta^{-b_\theta-2}}{(A_\theta^{-1} + b_\theta + 1)\Gamma(b_\theta + 1)}z^{b_\theta}(1+z)\exp(-z/A_\theta), \quad {z}>0,
\end{eqnarray*}
which is the density of a WL($b_\theta + 1,A_\theta^{-1})$.
\end{Proposition}
The proofs of propositions 2.2 and 2.3 are given in Appendix.

%The idea of multiplicative hazards framework can be extended to the bivariate case, considering that the  conditional hazard function for the $j$th individual in the $j$th cluster is given by
%\begin{equation*}
%    \lambda(t_{ij}\mid z_i,\mathbf{x}_i)=z_i \exp(\mathbf{x}_i^\top \beta)\lambda_0(t_{ij}), \quad j=1,2.
%\end{equation*}
%Under the assumption of conditional independence given $Z=z_i$, the bivariate conditional survival function is
%\begin{equation}
%S(t_{i1}, t_{i2}\mid z_i, \mathbf{x}_i)=\exp\left(z_i  \exp\left(\mathbf{x}_i^\top {\bm \beta}\right) \left[\Lambda_0(t_{i1})+\Lambda_0(t_{i2})\right]\right), \quad i=1,\ldots,n, \label{cond.S}
%\end{equation}
%Therefore, the marginal survival function and marginal pdf are given by
%\begin{align*}
%S(t_{i1}, t_{i2}\mid \mathbf{x}_i) &=\mathcal{L}_Z\left(\exp\left(\mathbf{x}_i^\top {\bm \beta}\right) \left[\Lambda_0(t_{i1})+\Lambda_0(t_{i2})\right]\right), \\
%f(t_{i1}, t_{i2}\mid \mathbf{x}_i) &=  \exp\left(\mathbf{x}_i^\top {\bm \beta}\right) \left[\lambda_0(t_{i1})+\lambda_0(t_{i2})\right] \mathcal{L}_Z^{(1)}\left(\exp\left(\mathbf{x}_i^\top {\bm \beta}\right) \left[\Lambda_0(t_{i1})+\Lambda_0(t_{i2})\right]\right), \quad i=1,\ldots,m,
%\end{align*}

%Some distributions considered in the literature in this bivariate context for the frailty models are gamma, Birnbaum-Saunders (BS), generalized exponential (GE), generalized inverse gamma (GIG) and recently, the WL model \citep{Tyagi21}.\\

\subsubsection{Shared WL frailty models}
\noindent

The idea of multiplicative hazards framework can be extended to the shared case, considering that the $i$th cluster has $n_i$ observations, for $i=1,\ldots,m$. In this case the conditional hazard and conditional survival functions for the $j$th individual in the $i$th cluster is given by
\begin{align*}
    \lambda(t_{ij}\mid z_i,\mathbf{x}_{ij})&=z_i \exp(\mathbf{x}_{ij}^\top \beta)\lambda_0(t_{ij}), \quad j=1,\ldots,n_i, \nonumber\\
    S(t_{ij}\mid z_i,\mathbf{x}_{ij})&=\exp\left(-z_i \exp(\mathbf{x}_{ij}^\top \beta)\Lambda_0(t_{ij})\right), \quad i=1,\ldots,m; \, j=1,\ldots,n_i,
\end{align*}
and with a similar development, it is obtained that the marginal survival and density functions are
\begin{align}
    S(t_{i1},\ldots, t_{in_i} \mid \mathbf{x}_{ij})&=\mathcal{L}_Z\left(\exp\left(\mathbf{x}_i^\top {\bm \beta}\right) \sum_{j=1}^{n_i}\Lambda_0(t_{ij})\right), \quad i=1,\ldots,m,  \label{mult.surv}\\
f(t_{i1},\ldots, t_{in_i} \mid \mathbf{x}_i) &=  (-1)^{n_i}\exp\left(\mathbf{x}_i^\top {\bm \beta}\right) \sum_{j=1}^{n_i}\lambda_0(t_{ij}) \mathcal{L}_Z^{(n_i)}\left(\exp\left(\mathbf{x}_i^\top {\bm \beta}\right) \sum_{j=1}^{n_i}\Lambda_0(t_{ij})\right). \label{mult.dens}
\end{align}
The particular case where $n_i=2$, $i=1,\ldots,m$ is known in the literature as the bivariate frailty model.
Distributions considered for the frailty terms $Z$ are the gamma \citep[with a general $n_i$; ][]{Clayton1978, Clayton1985} and the PS \citep[for the bivariate case, ][]{Manatunga1999}.
Other recent proposals are $Z_i\sim$ EE \citep{Barreto2019}, but, in this case, the authors present
$\mathcal{L}_Z^{(d)}$ for $d=2$ and $3$ and they claim ``analytical expressions for higher-order derivatives of $\mathcal{L}_Z^{(d)}$ can be obtained through programs such as Mathematica and Maple''. However, such derivatives also need to be programmed into some software and for $n_i$ moderately large this is impracticable. For instance, in our real data application it is observed up to $n_i=23$.
In a similar way, for $Z_i \sim$ GIG \cite{Piancastelli2021} presented $\mathcal{L}_Z^{(d)}$ in a recursive form, which is computationally inefficient when it is neither, again, moderately large.
For this reason, the WL appears as an alternative in this way. Taking advantage of the closed-form of the derivatives of the Laplace transform for the WL (see Proposition \ref{d.Lz}), for the first time we considered the WL in this shared frailty model context. With this, Eq. (\ref{mult.surv}) and (\ref{mult.dens}) are
\begin{eqnarray*}
S(t_{i1},\ldots, t_{in_i} \mid \mathbf{x}_{ij})&=&\left(1 + a_{\theta} \exp\left(\mathbf{x}_i^\top {\bm \beta}\right) \sum_{j=1}^{n_i}\Lambda_0(t_{ij})\right)^{-b_{\theta}-1}\left(1 + \frac{ \theta \exp\left(\mathbf{x}_i^\top {\bm \beta}\right) \sum_{j=1}^{n_i}\Lambda_0(t_{ij})}{2}\right),   \label{mult.surv.WL}\\
f(t_{i1},\ldots, t_{in_i} \mid \mathbf{x}_i) &=&  \exp\left(\mathbf{x}_i^\top {\bm \beta}\right) \sum_{j=1}^{n_i}\lambda_0(t_{ij}) \pi_{n_i}(b_{\theta}) a_{\theta}^{n_i-1} \left(1 + a_{\theta} \exp\left(\mathbf{x}_i^\top {\bm \beta}\right) \sum_{j=1}^{n_i}\Lambda_0(t_{ij})\right)^{-b_{\theta}-n_i-1}\nonumber\\
&~~~~~~~\times\left(1+\frac{\theta(\exp\left(\mathbf{x}_i^\top {\bm \beta}\right) \sum_{j=1}^{n_i}\Lambda_0(t_{ij})+n_i-1)}{\theta+2}\right), \quad i=1,\ldots,m, \label{mult.dens.WL}
\end{eqnarray*}
We remark that there are few models that allow obtaining a closed-form for these two functions for a general $n_i$: density and survival (marginal or unconditional in both cases).

\subsection{About $\Lambda_0(\cdot)$ in a WL frailty model context}
\noindent

For the univariate frailty WL model in (\ref{marg.surv.WL}), \cite{mota21} considered parametric models for $\Lambda_0(\cdot)$ taking the Weibull and Gompertz models. In a similar way, for the bivariate frailty WL model \cite{Tyagi21} considered the generalized Weibull and generalized log-logistic models. In this work, we considered a parametric model using the Weibull model with parametrization
\[
\Lambda_0(t)=\lambda \, t^\nu, \quad t, \lambda, \nu>0.
\]
However, in order to provide a more flexible scheme and for the first time in the literature, we also considered a non-parametric framework for $\Lambda_0(\cdot)$ in a frailty WL model.

\subsection{Kendall's $\tau$}
\noindent

\cite{Tyagi21} presented the Kendall's $\tau$ for the WL frailty model. However, such coefficient has a non-closer form and was presented for a different parameterization (not in terms of the variance for the frailty term as in here). In order to compare the WL and gamma frailty models, we present the Kendall's $\tau$ for the WL frailty model as
\begin{equation*}
    \tau=4 a_\theta (1+b_\theta) \int_0^{\infty} s\left(1+a_\theta s\right)^{-2b_\theta-4}\left(1+\frac{\theta s}{2}\right)\left(1+\frac{\theta(s+1)}{(\theta+2)}\right)ds-1,
\end{equation*}
whereas for the gamma frailty model it is well known that $\tau=\theta/(\theta+2)$ and for the IG frailty model $\tau=0.5-1/\theta+(2/\theta^2)\exp(2/\theta)E_1(2/\theta)$, where $E_1(\cdot)$ denotes the exponential integral function \citep[][page 228, eq. 5.1.1.]{Abramowitz:72}. The comparison between the three models is direct because $\theta$ represents the frailty variance in all the models. Figure \ref{kendall.WLgamma} compares the Kendall's $\tau$ for such models in terms of $\theta$. Note that, for a fixed variance for the frailty, the WL frailty model provides a greater Kendall's $\tau$ than the gamma and IG frailty models.

\begin{figure}[!htbp]
\centering
\includegraphics[width=5.5cm]{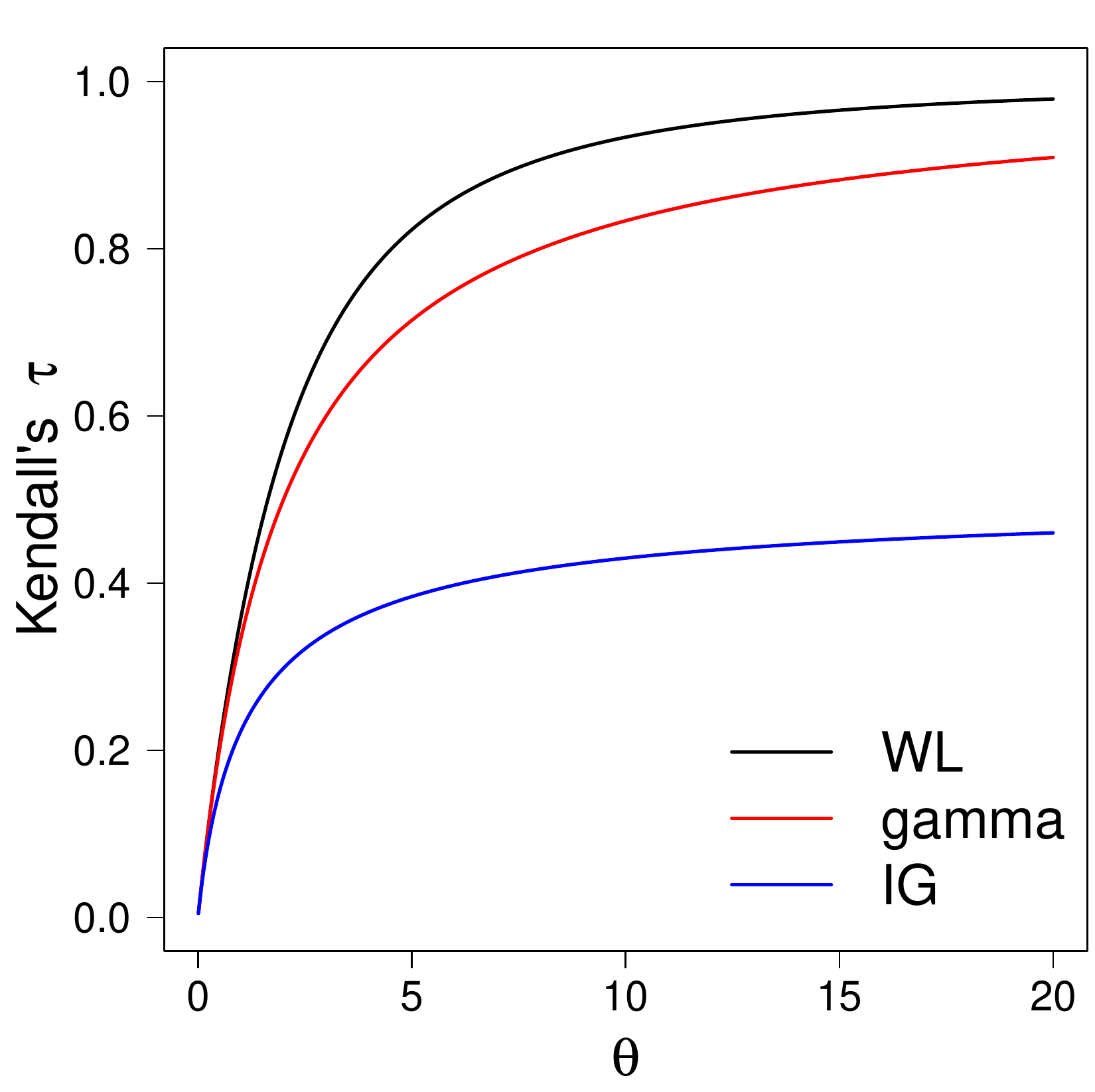}
\caption{Comparison among Kendall's $\tau$ for WL frailty, gamma frailty and IG frailty models.}
\label{kendall.WLgamma}
\end{figure}

\section{ML estimation for the WL frailty model}
\noindent

In this Section, we discuss the parameter estimation for the WL frailty model. First, we discuss an approach based on the assumption of a parametric model for the baseline distribution. Then, we present an EM algorithm to perform the parameter estimation using a non-parametric approach for the baseline distribution.

\subsection{Using a parametric approach for the baseline distribution}
\noindent

Let $Y_{ij}$ and $C_{ij}$ be the failure and censoring times for the $j$-th individual in the $i$-th group and $\textbf{x}_{ij}$ be a $p\times 1$ covariate vector (without intercept term), where $1\leq i \leq m$ and $1 \leq j \leq n_i$. Under a right censoring scheme, we observe the random variables $T_{ij}=\min(Y_{ij}, C_{ij})$ and $\delta_{ij}=I(Y_{ij} \leq C_{ij})$, where $I(A)=1$ if the event $A$ occurs (0 otherwise).
We assume the frailty terms $Z_1,\ldots,Z_m$ be a random sample from the WL$(\theta$) distribution. Considering the following assumptions:
\begin{itemize}
\item[i)] The pairs $(Y_{i1}, C_{i1}),\ldots,(Y_{in_i}, C_{in_i})$ are conditionally independent given $Z_i$, and $Y_{ij}$ and $C_{ij}$ are mutually independent for $j=1,\ldots,n_i$.
\item[ii)] $C_{i1}, \ldots, C_{in_i}$ are non-informative about $Z_i$.
\end{itemize}
Under this setting, the observed log-likelihood function is given by
\begin{align}
L({\bm \beta}, \Lambda_0, \theta)&=\prod_{i=1}^m\int_0^{+\infty} \prod_{j=1}^{n_i} \left[z_i \lambda_0(t_{ij})\exp\left(\textbf{x}_{ij}^\top {\bm \beta}\right)\right]^{\delta_{ij}} \exp\left(-z_i \Lambda_0(t_{ij})\textrm{e}^{\textbf{x}_{ij}^\top {\bm \beta}}\right)f(z_i;\theta) dz_i, \nonumber \\
&=\left(\frac{\theta a_{\theta}^{-\left(b_{\theta}+1\right)}}{2 \Gamma(b_{\theta})}\right)^m \exp\left(\sum_{i=1}^m \sum_{j=1}^{n_j} \delta_{ij} {\bm x}_{ij}^\top {\bm \beta}\right)\prod_{i=1}^m \int_0^{+\infty} z_i^{b_{\theta}^{(i)}-1}(1+z_i)\textrm{e}^{-z_i/a_{\theta}^{(i)}}dz_i \prod_{j=1}^{n_j} \left[\lambda_0(t_{ij})\right]^{\delta_{ij}}, \nonumber \\
&=\left(\frac{\theta a_{\theta}^{-\left(b_{\theta}+1\right)}}{2 \Gamma(b_{\theta})}\right)^m \exp\left(\sum_{i=1}^m \sum_{j=1}^{n_j} \delta_{ij} {\bm x}_{ij}^\top {\bm \beta}\right)\prod_{i=1}^m  \Gamma\left(b_{\bm \psi}^{(i)}\right)\left(a_{\bm \psi}^{(i)}\right)^{b_{\bm \psi}^{(i)}}\left(1+a_{\bm \psi}^{(i)} b_{\bm \psi}^{(i)}\right) \prod_{j=1}^{n_j} \left[\lambda_0(t_{ij})\right]^{\delta_{ij}}, \label{observed.loglike}
\end{align}
where $a_{\bm \psi}^{(i)}=\left(\sum_{j=1}^{n_i}\Lambda_0(t_{ij})\textrm{e}^{{\bm x}_{ij}^\top {\bm \beta}}+1/a_{\theta}\right)^{-1}$ and $b_{\bm \psi}^{(i)}=r_i+b_{\theta}$.
In a parametric approach, $\Lambda_0(t)$ or $\lambda_0(t)$ are specified by a set of parameters, say ${\bm \nu}$, and then the parameter vector is reduced to $({\bm \beta}, {\bm \nu}, \theta)$.
For instance, for the Weibull (WEI) distribution we use the parameterization $\Lambda_0(t)=\lambda\, t^\rho$ and $\lambda_0(t)=\lambda \, \rho\, t^{\rho-1}$, where $t>0$ and ${\bm \nu}=(\lambda,\rho)\in \mathbb{R}_+^2$.
From a classical approach the ML estimator can be obtained maximizing $\log L({\bm \beta}, {\bm \nu}, \theta)$ in relation to ${\bm \beta}, {\bm \nu}$ and $\theta$. %However, it can be more attractive to discuss a non-parametric approach for the baseline distribution. For this, we consider an estimation procedure based on the EM algorithm.

\subsection{ML estimation via EM-algorithm and a semiparametric approach}\label{sec.EM}
\noindent

In this subsection, we discuss an EM-type algorithm to perform the parameter estimation for the WL frailty model. Despite this approach can also be used for a parametric approach, the motivation to this development arises from the application of a non-parametric approach to the failure times.

For our particular problem, the complete data are given by ${\bm D}_{c}=({\bm t}^\top, {\bm \delta}^\top, \textbf{X}^\top, {\bm Z}^\top)$, where ${\bm t}^\top=(t_{11},\ldots, t_{mn_m})$, ${\bm \delta}^\top=(\delta_{11},\ldots, \delta_{mn_m})$, $\textbf{X}^\top=(x_{11},\ldots,x_{mn_m})$ and ${\bm Z}^\top=(Z_1,\ldots,Z_m)$. In our context, ${\bm D}_{obs}=\left({\bm t}^\top, {\bm \delta}^\top, \textbf{X}^\top\right)$ represents the observed data and ${\bm Z}$ (the frailty terms) denotes the latent variables. The complete likelihood function can be written conveniently as $L({\bm \beta},{\bm \Lambda}_0,\theta;{\bm D}_{c})=L_1({\bm \beta}, \Lambda_0;{\bm D}_{c})\times L_2(\theta; {\bm Z})$, where $L_1({\bm \beta}, \Lambda_0;{\bm D}_{c})=\prod_{i=1}^m \prod_{j=1}^{n_i} \left[z_i \lambda_0(t_{ij})\exp(\textbf{x}_{ij}^\top {\bm \beta})\right]^{\delta_{ij}}\exp(-z_i \Lambda_0(t_{ij})\textrm{e}^{\textbf{x}_{ij}^\top {\bm \beta}})$ and $L_2(\theta; {\bm Z})=\prod_{i=1}^m f(z_i;\theta)$.\\
\noindent The complete log-likelihood function is given by $\ell_c({\bm \beta}, {\bm \Lambda}_0,\theta;{\bm D}_{c})=\ell_{1c}({\bm \beta}, {\bm \Lambda}_0;{\bm D}_{c})+\ell_{2c}(\theta; {\bm Z})$, where
\begin{align*}
\ell_{1c}({\bm \beta}, {\bm \Lambda}_0; {\bm D}_{c})&= \sum_{i=1}^m \sum_{j=1}^{n_i} \left[\delta_{ij}\left(\textbf{x}_{ij}^\top {\bm \beta}+\log \lambda_0(t_{ij})\right)-Z_i \Lambda_0(t_{ij})\textrm{e}^{\textbf{x}_{ij}^\top {\bm \beta}}\right]\\
\ell_{2c}(\theta; {\bm Z})&=m \left[\log(\theta) -\log \Gamma(b_{\theta})-(b_\theta+1)\log(a_\theta) \right]-(b_\theta-1)\sum_{i=1}^m\log(Z_i)-\frac{1}{a_\theta} \sum_{i=1}^m Z_i.
\end{align*}
Let ${\bm \psi}^{(\ell)}=({\bm \beta}^{(\ell)},\Lambda_0^{(\ell)}, \theta^{(\ell)})$ be the estimate of ${\bm \psi}$ at the $\ell$-th iteration and denote $Q({\bm \psi}\mid {\bm \psi}^{(\ell)})$ as the conditional expectation of $\ell_c({\bm \psi})$ given the observed data and ${\bm \psi}^{(\ell)}$. With these notations, we have that $Q({\bm \psi}\mid {\bm \psi}^{(\ell)})=Q_1(({\bm \beta},\Lambda_0)\mid {\bm \psi}^{(\ell)})+Q_2(\theta \mid {\bm \psi}^{(\ell)})$, where
\begin{align*}
Q_1(({\bm \beta},\Lambda_0)\mid {\bm \psi}^{(\ell)})&= \sum_{i=1}^m \sum_{j=1}^{n_i} \left[\delta_{ij}\left(\textbf{x}_{ij}^\top {\bm \beta}+\log[\lambda_0(t_{ij})\right)]-\widehat{z}_i^{(\ell)} \Lambda_0(t_{ij})\textrm{e}^{\textbf{x}_{ij}^\top {\bm \beta}}\right]\\
Q_2(\theta \mid {\bm \psi}^{(\ell)})&=m \left[\log(\theta) -\log \Gamma(b_{\theta})-(b_\theta+1)\log(a_\theta) \right]-(b_\theta-1)\sum_{i=1}^m\widehat{\kappa}_i^{(\ell)}-\frac{1}{a_\theta} \sum_{i=1}^m \widehat{z}_i^{(\ell)},
\end{align*}
where $\widehat{z}^{(\ell)}_i=E(Z_i\mid {\bm D}_{obs}, {\bm \psi}={\bm \psi}^{(\ell)})$ and $\widehat{\kappa}^{(\ell)}_i=E(\log(Z_i)\mid {\bm D}_{obs}, {\bm \psi}={\bm \psi}^{(\ell)})$. It is possible to show that $Z_i \mid t_{i1},\ldots,t_{in_i},\delta_{i1},\ldots,\delta_{in_i} \sim WL\left(\left[a_{\bm \psi}^{(i)}\right]^{-1}, b_{\bm \psi}^{(i)}\right)$.
Therefore, using the results presented in Section \ref{sec2.1} it is immediate that for $i=1,\ldots,m$,
\begin{align}
\widehat{z}_i^{(\ell)}&=\frac{b_{\widehat{\bm \psi}^{(\ell)}}^{(i)} \left(\left[a_{\widehat{\bm \psi}^{(\ell)}}^{(i)}\right]^{-1}+b_{\widehat{\bm \psi}^{(\ell)}}^{(i)}+1\right)}{\left[a_{\widehat{\bm \psi}^{(\ell)}}^{(i)}\right]^{-1}\left(\left[a_{\widehat{\bm \psi}^{(\ell)}}^{(i)}\right]^{-1}+b_{\widehat{\bm \psi}^{(\ell)}}^{(i)}\right)}, \qquad \mbox{and} \label{widez}\\
\widehat{\kappa}_i^{(\ell)}&=-\frac{\left[a_{\widehat{\bm \psi}^{(\ell)}}^{(i)}\right]^{-1}}{b_{\widehat{\bm \psi}^{(\ell)}}^{(i)} \left(\left[a_{\widehat{\bm \psi}^{(\ell)}}^{(i)}\right]^{-1}+b^{(i)}_{\widehat{\bm \psi}^{(\ell)}}\right)}+\Psi\left(b^{(i)}_{\widehat{\bm \psi}^{(\ell)}}+1\right)+\log\left[a^{(i)}_{\widehat{\bm \psi}^{(\ell)}}\right]\label{widekappa}.
\end{align}

On the other hand, using the traditional development of the Cox model, it is possible to construct a discrete version of the cumulative baseline hazard function, replacing $\Lambda_0(t)$ by $\Lambda_0^D(t)=\sum_{k: t_{(k)}\leq t} \lambda_0(t_{(k)})$, where $t_{(1)},\ldots,t_{(q)}$ denotes the ordered distinct failure times $t_{ij}'s$ and $q$ is the number of different observed failure times.
With this, the $Q_1$ function is given by
\begin{align*}
Q_1(({\bm \beta},\Lambda_0)\mid {\bm \psi}^{(\ell)})&= \sum_{k=1}^q d_{(k)} \log \left[\lambda_0(t_{(k)})\right]+\sum_{i=1}^m \sum_{j=1}^{n_i} \delta_{ij}\textbf{x}_{ij}^\top {\bm \beta}-\sum_{k=1}^q \lambda_0(t_{(k)})\sum_{i,j \in R(t_{(k)})} \widehat{z}_i^{(\ell)} \textrm{e}^{\textbf{x}_{ij}^\top {\bm \beta}},
\end{align*}
where $R(t_{(k)})=\{(i,j): t_{ij}>t_{(k)}\}$ are the observations in risk at the time $t_{(k)}$ and $d_{(k)}$ denotes the number of failures at $t_{(k)}$, for $k=1,\ldots,q$. Note that the solution for $\lambda_0(t_{(k)})$ is given by
\begin{equation*}
\widehat{\lambda_0}(t_{(k)})=\frac{d_{(k)}}{\sum_{i,j \in R(t_{(k)})}\exp\left(\textbf{x}_{ij}^\top {\bm \beta}+\log \widehat{z}^{(\ell)}_i\right)}.
\end{equation*}
With this result, the expression for $Q_1$ is reduced to
\begin{align*}
Q_1(({\bm \beta},\Lambda_0)\mid {\bm \psi}^{(\ell)})&=-\sum_{k=1}^q d_{(k)}\log \left(\sum_{i,j \in R(t_{(k)})} \exp\left(\textbf{x}_{ij}^\top{\bm \beta}+\log \widehat{z}_i^{(\ell)}\right)\right)+\sum_{i=1}^m\sum_{j=1}^{n_i} \delta_{ij} \textbf{x}_{ij}^\top {\bm \beta}.
\end{align*}
Note that $Q_1(\cdot)$ has the same form of the partial log-likelihood function of the Cox model, except for the offset $\log \widehat{z}_i^{(\ell)}$. For this, to update ${\bm \beta}$ in the M-step we can use the Cox approach. Finally, the non-parametric estimator for $\Lambda_0(\cdot)$ in the $\ell$-th step of the algorithm is given by
\begin{align*}
\widehat{\Lambda}^{(\ell)}_0(t)=\sum_{k: t_{(k)}\leq t} \frac{d_{(k)}}{\sum_{i,j \in R(t_{(k)})}\exp\left(\textbf{x}_{ij}^\top \widehat{\bm \beta}^{(\ell)}+\log \widehat{z}^{(\ell)}_i\right)}, \quad t>0.
\end{align*}
In short, the EM algorithm is summarized as follows
\begin{itemize}
\item \textbf{E-step:} For $i=1,\ldots, m$, update $\widehat{z}_i^{(\ell+1)}$ and $\widehat{\kappa}_i^{(\ell+1)}$ using Equations (\ref{widez}) and (\ref{widekappa}) with $\widehat{\bm \beta}^{(\ell)}$, $\widehat{\Lambda}^{(\ell)}_0(\cdot)$ and $\widehat{\theta}^{(\ell)}$, the parameters at the $\ell$-th iteration.
\item \textbf{M1-step}: Update $\widehat{\bm \beta}^{(\ell+1)}$ and $\Lambda_0^{(\ell+1)}(\cdot)$ by fitting a Cox regression model with offset $\log \widehat{z}_i^{(\ell+1)}$.
\item \textbf{M2-step}: Update $\widehat{\theta}^{(\ell+1)}$ by maximizing $Q_2(\theta \mid {\bm \psi}^{(\ell)})$ in relation to $\theta$.
\end{itemize}
The E-, M1- and M2-steps are iterated until a convergence criterion is satisfied. For instance,
we consider $\left|\left(\widehat{\bm \beta}^{(\ell+1)},\widehat{\theta}^{(\ell+1)}\right)-\left(\widehat{\bm \beta}^{(\ell)},\widehat{\theta}^{(\ell)}\right)\right|<\epsilon$, where $\epsilon$ is a predefined value. Initial values $\widehat{\bm \beta}^{(0)}$ and $\widehat\Lambda_0^{(0)}(\cdot)$ can be obtained based on the usual Cox regression model. In addition, it is possible to fix $\widehat{z}_i^{(0)}=1$ and $\widehat{\kappa}_i^{(0)}=0$, for $i=1,\ldots,m$, and an arbitrary value for $\theta$. For instance, we use $\widehat{\theta}^{(0)}=0.5$. Standard errors for $\widehat{\bm \beta}$ and $\widehat{\theta}$ can be obtained following the suggestion of \cite{Klein92}. For this, it is considered a profile log-likelihood function, say $\ell({\bm \beta},\theta)$, replacing $\Lambda_0$ by its estimation $\widehat{\Lambda}_0$ and taking the logarithm in Equation (\ref{observed.loglike}). Therefore, the information matrix is computed as $I({\bm \beta},\theta)=-\partial^2 \ell({\bm \beta},\theta)/\partial \left({\bm \beta}, \theta\right)\partial^\top \left({\bm \beta}, \theta\right)$. The variance of the ML estimators, say $\widehat{\bm \beta}$ and $\widehat{\theta}$, can be estimated numerically. For simplicity's sake, we omit such details.

\begin{Remark}
We also implement an EM-type algorithm when the WEI distribution is assumed as the baseline model. In this case, the EM algorithm is essentially the same. However, instead of a Cox regression model, we need to perform a WEI regression model in the M1-step, but considering the $\log\left(\widehat{z}_i^{(\ell)}\right)$ as offset. We implement this using the function \texttt{survfit} included in the package \texttt{survival} \citep{survival-package} of \cite{R2022}.
\end{Remark}

\begin{Remark}
The package \texttt{extrafrail} \citep{extrafrail-package} of \cite{R2022} included the computational implementation for the WL frailty model considering as baseline model the WEI distribution and the semi-parametric specification. For instance, to fit the non-parametric case, it can be used
\begin{verbatim}
frailtyWL(formula, data, dist = "np")
\end{verbatim}
where as is usually in survival analysis with random effects in R, \texttt{formula} can defined as
\begin{verbatim}
Surv(time, event) ~ covariates + cluster(id)
\end{verbatim}
A similar syntax can be used to fit the Weibull case specifying \texttt{dist="weibull"} in last sentence.
\end{Remark}

\newpage

\section{Simulation study}

\noindent In this Section we present two simulation studies related to the WL frailty model with a semiparametric baseline. The first study is devoted to study the recovery parameters under different scenarios and the second study assesses the performance of the model with a misspecification in the frailty distribution.

\subsection{Recovery parameters}
\noindent In this subsection, we study the properties of the ML estimators in finite samples obtained using the EM algorithm discussed in subsection \ref{sec.EM}. The data were drawn from a similar scenery than the application. We considered $\Lambda_0(t)=\lambda \, t^\rho$, i.e., the Weibull distribution. We fixed three cases: i) mean $\mu_w=8.6$ and variance $\sigma_w^2=230$, which implies $\lambda\approx 5.6976$ and $\rho\approx 0.5985$; ii) $\mu_w=6.0$ and $\sigma_w^2=230$, resulting in $\lambda \approx 2.5319$ and $\rho \approx 0.4593$ and; iii) $\mu_w=8.6$ and $\sigma_w^2=100$, resulting in $\lambda \approx 7.9786$ and $\rho \approx 0.8630$. For the clusters, we assumed three  scenarios.
\begin{itemize}
\item Case 1. $m=396$ clusters with the following distribution: 200 clusters with 1 observation, 100 clusters with 2 observations, 50 clusters with 3 observations, 20 clusters with 4 observations, 20 clusters with 5 observations and 6 clusters with 10 observations, totalizing 790 observations.
\item Case 2. $m=396$ clusters with the following distribution: 200 clusters with 2 observations, 100 clusters with 4 observations, 50 clusters with 6 observations, 20 clusters with 8 observations, 20 clusters with 10 observations and 6 clusters with 20 observations, totalizing 1,580 observations.
\item Case 3. $m=792$ clusters with the following distribution: 400 clusters with 1 observation, 200 clusters with 2 observations, 100 clusters with 3 observations, 40 clusters with 4 observations, 40 clusters with 5 observations and 12 clusters with 10 observations, totalizing 1,580 observations.
\end{itemize}
Note that we are considering clusters with different number of observations. We also highlight that Case 2 doubles the observations for each cluster (keeping the number of clusters) and Case 3 doubles the number of clusters (keeping the number of observations for each cluster). Frailty terms were drawn from the WL$(\theta)$ model, where three values were used for $\theta: 0.1, 0.25$ and $0.5$. We
assume the multiplicative hazard model in (\ref{cond.lambda}) with four covariates, one of them simulated from the categorical distribution with probabilities $(0.4,0.4$ and $0.2)$. The other three covariates were drawn from the Bernoulli distribution with success probabilities of $0.7$, $0.6$ and $0.5$, respectively. Therefore, for each individual the covariate vector is $\bf{x}_i^\top=(x_{11i},x_{12i},x_{2i},x_{3i},x_{4i})$. For each case, we consider ${\bm \beta}=(\beta_{11},\beta_{12},\beta_{2},\beta_{3},\beta_{4})=(0.3,1.1,0.4,-0.5,-0.3)$.
We also consider three percentages of censoring data: 0\%, 10\% and 25\%. With this scheme and defining $\xi_i=\bf{x}_i^\top {\bm \beta}$, we obtain
\[
\Lambda_0(t\mid z_i, {\bf{x}}_i)=\lambda \, z_i \, \textrm{e}^{\xi_i} t^\rho,
\]
i.e., $t\mid z_i, \bf{x}_i$ has Weibull distribution with parameters $\lambda_i^*=\lambda \, z_i \, \textrm{e}^{\xi_i}$ and $\rho^*=\rho$, where it is simple to draw values. In order to obtain a $100\times q$\% of censoring times, we fix $C_i$ (the censoring times) as the $100\times (1-q)$-th percentile of the last referred Weibull distribution, so that $P(T_i>C_i)=q$, as requested. For each of the 81 combinations of $(\mu_2, \sigma_w^2)$, clusters, $\theta$ and censoring, we draw 1,000 samples and then, we compute the ML estimators using the EM algorithm implemented in the \texttt{extrafrail} package. Table \ref{sim.1} summarizes the estimated bias (bias), the mean of the estimated standard errors (SE) and the root of the estimated mean squared error (RMSE) for two cases of the 0\% censoring. The remaining cases are given as supplementary material. In general terms, for all the cases the bias is acceptable and the terms SE and RMSE are closer, suggesting that the estimated standard errors for all the estimators are well estimated. Results are similar for different combinations of $\mu_w$ and $\sigma_w^2$. We also note that the frailty variance is better estimated when the number of clusters is maintained, but the observations in each cluster is augmented than the case when the observations in each cluster in maintained, but the number of clusters is augmented. In simple words, to estimate the frailty variance is better to have more intra clusters observations than inter clusters, as expected.

\begin{landscape}
\begin{table}[!h]
\caption{Estimated bias, SE and RMSE for the shared WL frailty semiparametric model (censoring: 0\%).}\label{sim.1}
\resizebox{\linewidth}{!}{
\begin{tabular}{cccrrrrrrrrr}
\hline
           &            &            & \multicolumn{ 3}{c}{$\mu_w=8.6, \sigma_w^2=230$} &  \multicolumn{ 3}{c}{$\mu_w=6.0, \sigma_w^2=230$} & \multicolumn{ 3}{c}{$\mu_w=8.6, \sigma_w^2=100$} \\

     $\theta$ &       $m/n_j$ &  estimator &       bias &        SE &       RMSE &       bias &        SE &       RMSE &       bias &        SE &       RMSE \\
\hline

      0.10 &     case i &     $\widehat{\beta}_{11}$ &    $-$0.0238 &     0.0761 &     0.0894 &    $-$0.0270 &     0.0761 &     0.0879 &    $-$0.0291 &     0.0754 &     0.0892 \\

           &            &     $\widehat{\beta}_{12}$ &    $-$0.0114 &     0.1020 &     0.1206 &    $-$0.0149 &     0.1020 &     0.1173 &    $-$0.0205 &     0.1015 &     0.1151 \\

           &            &      $\widehat{\beta}_{2}$ &     0.0029 &     0.0819 &     0.0826 &     0.0056 &     0.0820 &     0.0798 &     0.0016 &     0.0815 &     0.0815 \\

           &            &      $\widehat{\beta}_{3}$ &    $-$0.0227 &     0.0749 &     0.0823 &    $-$0.0231 &     0.0749 &     0.0836 &    $-$0.0248 &     0.0744 &     0.0835 \\

           &            &      $\widehat{\beta}_{4}$ &     0.0214 &     0.0694 &     0.0819 &     0.0239 &     0.0694 &     0.0803 &     0.0192 &     0.0690 &     0.0818 \\

           &            &      $\widehat{\theta}$ &     0.0210 &     0.0379 &     0.0339 &     0.0207 &     0.0382 &     0.0350 &     0.0204 &     0.0361 &     0.0343 \\
\cline{2-12}
           &    case ii &     $\widehat{\beta}_{11}$ &    $-$0.0012 &     0.0538 &     0.0627 &     0.0007 &     0.0538 &     0.0595 &    $-$0.0012 &     0.0534 &     0.0585 \\

           &            &     $\widehat{\beta}_{12}$ &     0.0229 &     0.0690 &     0.0803 &     0.0281 &     0.0691 &     0.0823 &     0.0281 &     0.0688 &     0.0817 \\

           &            &      $\widehat{\beta}_{2}$ &     0.0087 &     0.0566 &     0.0583 &     0.0100 &     0.0566 &     0.0571 &     0.0090 &     0.0563 &     0.0582 \\

           &            &      $\widehat{\beta}_{3}$ &    $-$0.0027 &     0.0520 &     0.0574 &    $-$0.0019 &     0.0521 &     0.0545 &    $-$0.0043 &     0.0515 &     0.0552 \\

           &            &      $\widehat{\beta}_{4}$ &    $-$0.0010 &     0.0481 &     0.0525 &    $-$0.0019 &     0.0482 &     0.0530 &     0.0008 &     0.0478 &     0.0547 \\

           &            &      $\widehat{\theta}$ &     0.0070 &     0.0236 &     0.0200 &     0.0072 &     0.0237 &     0.0197 &     0.0058 &     0.0228 &     0.0194 \\
\cline{2-12}
           &   case iii &     $\widehat{\beta}_{11}$ &    $-$0.0200 &     0.0553 &     0.0641 &    $-$0.0227 &     0.0553 &     0.0640 &    $-$0.0198 &     0.0548 &     0.0650 \\

           &            &     $\widehat{\beta}_{12}$ &     0.0148 &     0.0706 &     0.0781 &     0.0162 &     0.0706 &     0.0796 &     0.0150 &     0.0702 &     0.0809 \\

           &            &      $\widehat{\beta}_{2}$ &     0.0019 &     0.0578 &     0.0579 &     0.0066 &     0.0578 &     0.0582 &     0.0015 &     0.0575 &     0.0613 \\

           &            &      $\widehat{\beta}_{3}$ &    $-$0.0092 &     0.0529 &     0.0552 &    $-$0.0091 &     0.0529 &     0.0594 &    $-$0.0083 &     0.0525 &     0.0567 \\

           &            &      $\widehat{\beta}_{4}$ &     0.0000 &     0.0492 &     0.0556 &     0.0052 &     0.0492 &     0.0534 &     0.0035 &     0.0488 &     0.0525 \\

           &            &      $\widehat{\theta}$ &     0.0225 &     0.0266 &     0.0304 &     0.0231 &     0.0268 &     0.0308 &     0.0202 &     0.0254 &     0.0282 \\
\hline
      0.25 &     case i &     $\widehat{\beta}_{11}$ &    $-$0.0216 &     0.0813 &     0.0862 &    $-$0.0215 &     0.0815 &     0.0884 &    $-$0.0178 &     0.0803 &     0.0866 \\

           &            &     $\widehat{\beta}_{12}$ &    $-$0.0602 &     0.1087 &     0.1257 &    $-$0.0571 &     0.1087 &     0.1321 &    $-$0.0633 &     0.1078 &     0.1287 \\

           &            &      $\widehat{\beta}_{2}$ &    $-$0.0402 &     0.0881 &     0.0941 &    $-$0.0403 &     0.0882 &     0.0960 &    $-$0.0369 &     0.0871 &     0.0915 \\

           &            &      $\widehat{\beta}_{3}$ &     0.0459 &     0.0798 &     0.0947 &     0.0508 &     0.0799 &     0.0949 &     0.0483 &     0.0788 &     0.0931 \\

           &            &      $\widehat{\beta}_{4}$ &    $-$0.0288 &     0.0746 &     0.0821 &    $-$0.0270 &     0.0746 &     0.0812 &    $-$0.0243 &     0.0737 &     0.0807 \\

           &            &      $\widehat{\theta}$ &    $-$0.0208 &     0.0441 &     0.0420 &    $-$0.0198 &     0.0444 &     0.0408 &    $-$0.0266 &     0.0414 &     0.0447 \\
\cline{2-12}
           &    case ii &     $\widehat{\beta}_{11}$ &     0.0152 &     0.0574 &     0.0609 &     0.0138 &     0.0575 &     0.0645 &     0.0143 &     0.0568 &     0.0628 \\

           &            &     $\widehat{\beta}_{12}$ &    $-$0.0302 &     0.0740 &     0.0834 &    $-$0.0290 &     0.0741 &     0.0830 &    $-$0.0291 &     0.0735 &     0.0839 \\

           &            &      $\widehat{\beta}_{2}$ &    $-$0.0088 &     0.0602 &     0.0591 &    $-$0.0085 &     0.0603 &     0.0598 &    $-$0.0069 &     0.0597 &     0.0602 \\

           &            &      $\widehat{\beta}_{3}$ &     0.0079 &     0.0554 &     0.0568 &     0.0069 &     0.0555 &     0.0588 &     0.0077 &     0.0547 &     0.0583 \\

           &            &      $\widehat{\beta}_{4}$ &     0.0276 &     0.0513 &     0.0625 &     0.0294 &     0.0514 &     0.0617 &     0.0280 &     0.0508 &     0.0616 \\

           &            &      $\widehat{\theta}$ &     0.0014 &     0.0333 &     0.0281 &     0.0005 &     0.0334 &     0.0265 &    $-$0.0026 &     0.0320 &     0.0269 \\
\cline{2-12}
           &   case iii &     $\widehat{\beta}_{11}$ &    $-$0.0164 &     0.0586 &     0.0630 &    $-$0.0146 &     0.0587 &     0.0627 &    $-$0.0133 &     0.0580 &     0.0605 \\

           &            &     $\widehat{\beta}_{12}$ &    $-$0.0456 &     0.0752 &     0.0907 &    $-$0.0488 &     0.0752 &     0.0921 &    $-$0.0484 &     0.0747 &     0.0907 \\

           &            &      $\widehat{\beta}_{2}$ &    $-$0.0238 &     0.0613 &     0.0641 &    $-$0.0216 &     0.0614 &     0.0616 &    $-$0.0233 &     0.0609 &     0.0628 \\

           &            &      $\widehat{\beta}_{3}$ &    $-$0.0037 &     0.0561 &     0.0557 &     0.0009 &     0.0562 &     0.0584 &     0.0000 &     0.0555 &     0.0553 \\

           &            &      $\widehat{\beta}_{4}$ &     0.0044 &     0.0522 &     0.0552 &     0.0055 &     0.0523 &     0.0561 &     0.0039 &     0.0518 &     0.0531 \\

           &            &      $\widehat{\theta}$ &    $-$0.0282 &     0.0313 &     0.0400 &    $-$0.0279 &     0.0315 &     0.0398 &    $-$0.0310 &     0.0296 &     0.0415 \\
\hline
\end{tabular}
}
\end{table}

\end{landscape}

\subsection{Assessing the baseline distribution}
\noindent

We devote this subsection to assessing the estimation of the baseline distribution when the frailty WL model is used with a non-parametric specification for $\Lambda_0(\cdot)$. For this we draw the covariates using the same specification as the last study. The frailty terms were drawn from the WL model with mean 1 and variance equal to $0.1$, $0.25$ and $0.5$. In addition, we also consider three different distributions, implying a misspecification problem for the frailty distribution. We consider $W$ with uniform distribution between 0 and 2; with distribution gamma with shape and rate equal to 1; and with log-normal distribution with location $-\log(2)/2$ and scale $\log(2)$. Those cases imply mean 1 for the distribution and variance $1/3$, $1$ and $1$, respectively. For all the cases, we consider the true baseline coming from the Weibull model with the same three specifications previously used: i) mean $\mu_w=8.6$ and variance $\sigma_w^2=230$; ii) $\mu_w=6.0$ and $\sigma_w^2=230$ and; iii) $\mu_w=8.6$ and $\sigma_w^2=100$. In order to assess the estimation for the non-parametric baseline estimator, we use two measures: mean and median. Theoretical values (well known $\mu_w=\lambda \Gamma(1+1/\rho)$ and $\xi_w=\lambda [\log(2)]^{1/\rho}$ for mean and median, respectively) are compared with the respective values based on $\widehat{\Lambda}_0(\cdot)$. As for positive random variables it is valid that $E(T)=\int_0^\infty S(t)dt$, we can use the area under curve of $\widehat{S}_0(t)=\exp(-\widehat{\Lambda}_0(t))$ to estimate the mean of the baseline distribution as
\begin{align*}
\widehat{\mu_w}&=t_{(1)}+\sum_{i=1}^{k-1} \left(t_{(j+1)}-t_{(j)}\right)\widehat{S}_0(t_{(j)}),
\end{align*}
where $t_{(1)} < t_{(2)} < \cdots < t_{(k)}$ denote the different failure times observed in the sample. On the other hand, to estimate the median of the distribution (say $\widehat{\xi}_w$) we use a linear interpolation between $\left(\tau_{k^*},\widehat{S}_0(\tau_{k^*})\right)$ and $\left(\tau_{k^*+1},\widehat{S}_0(\tau_{k^*+1})\right)$, where $k^*=\{i\in \{1,\ldots,k\}: \widehat{S}_0(t_{i})\leq 0.5 \wedge \widehat{S}_0(t_{i+1})\geq 0.5\}$. The results are summarized considering the bias (B) and relative bias (RB) with their corresponding standard deviation (say B-SD and RB-SD, respectively) for the three measures, i.e.,
\begin{align*}
\mbox{B}(\varphi)&=\sum_{j=1}^{1000} (\widehat{\varphi}_i-\varphi), \quad \mbox{RB}(\varphi)=\sum_{j=1}^{1000} \left(\frac{\widehat{\varphi}_i-\varphi}{\varphi}\right),\\
\mbox{B-SD}(\varphi)&=\sum_{j=1}^{1000} (\widehat{\varphi}_i-\varphi)^2 \quad \mbox{and} \quad \mbox{RB-SD}(\varphi)=\sum_{j=1}^{1000} \left(\frac{\widehat{\varphi}_i-\varphi}{\varphi}\right)^2,
\end{align*}
for $\varphi \in \{\mu_w, \xi_w\}$. Table \ref{sim.2} summarizes the results for the case with 0\% of censoring. The cases with 10\% and 25\% of censoring are presented as supplementary material. Note that when the frailty distribution is well specified, the maximum RB attaches 0.045, 0.207 and 0.387 for $\theta=0.1$, $0.25$ and $0.5$, respectively, for the mean and a similar pattern is obtained for the mean. Therefore, an increment in the frailty variance deteriorates the estimation for the baseline distribution. Specifically, the baseline distribution is overestimated with increasing frailty variance. This overestimation is also illustrated in Figure \ref{np.versus.true}. On the other hand, for the cases where the frailty distribution was misspecified, the RB ranges from 0.223 to 1.552 and therefore, a misspecification in the frailty distribution also helps to overestimate the baseline distribution. Finally, when the distribution is well specified, a decrement in the mean or in the standard deviation of the true baseline distribution reduces the bias of the baseline estimator or at least makes it more homogeneous.

\begin{table}[!htp]
\caption{Estimated bias for mean, median and standard deviation of the baseline distribution using shared WL frailty semiparametric model (censoring: 0\%).}\label{sim.2}
\resizebox{\textwidth}{!}{
%\begin{footnotesize}
%\begin{adjustbox}{angle=90}
\begin{tabular}{rrrrrrrrrrrrrrrrrrrrrrrrrrr}
\hline
           &            &            &                           \multicolumn{ 4}{c}{WL ($\theta=0.1$)} &                           \multicolumn{ 4}{c}{WL ($\theta=0.25$)} &                           \multicolumn{ 4}{c}{WL ($\theta=0.5$)} &                            \multicolumn{ 4}{c}{U ($\theta=0.33$)} &                        \multicolumn{ 4}{c}{gamma ($\theta=1$)} &                           \multicolumn{ 4}{c}{LN ($\theta=1$)} \\

      $\mu_w - \sigma_w^2$ &        m/n &    measure &          B &       B-SD &         RB &      RB-SD &          B &       B-SD &         RB &      RB-SD &          B &       B-SD &         RB &      RB-SD &          B &       B-SD &         RB &      RB-SD &          B &       B-SD &         RB &      RB-SD &          B &       B-SD &         RB &      RB-SD \\
\hline

         8.6 - 230 &        case i &         $\mu_w$ &      0.139 &      1.238 &      0.016 &      0.144 &      0.931 &      1.331 &      0.108 &      0.155 &      2.519 &      1.547 &      0.293 &      0.180 &      3.444 &      1.706 &      0.400 &      0.198 &      6.399 &      2.034 &      0.744 &      0.236 &      3.657 &      1.731 &      0.425 &      0.201 \\

           &            &         $\xi_w$ &      0.077 &      0.496 &      0.025 &      0.161 &      0.137 &      0.494 &      0.044 &      0.160 &      0.524 &      0.591 &      0.170 &      0.191 &      1.018 &      0.659 &      0.329 &      0.213 &      1.639 &      0.840 &      0.531 &      0.272 &      1.235 &      0.736 &      0.400 &      0.238 \\

%           &            &     $\sigma^2_w$ &     -2.190 &      1.427 &     -0.144 &      0.094 &     -1.193 &      1.416 &     -0.079 &      0.093 &      0.298 &      1.410 &      0.020 &      0.093 &      0.689 &      1.423 &      0.045 &      0.094 &      2.859 &      1.385 &      0.189 &      0.091 &      1.094 &      1.471 &      0.072 &      0.097 \\
\cline{2-27}
           &        case ii &         $\mu_w$ &      0.367 &      0.851 &      0.043 &      0.099 &      1.419 &      0.938 &      0.165 &      0.109 &      2.584 &      1.067 &      0.300 &      0.124 &      2.893 &      1.113 &      0.336 &      0.129 &      5.689 &      1.329 &      0.662 &      0.155 &      3.930 &      1.153 &      0.457 &      0.134 \\

           &            &         $\xi_w$ &      0.177 &      0.361 &      0.057 &      0.117 &      0.460 &      0.403 &      0.149 &      0.130 &      0.742 &      0.436 &      0.240 &      0.141 &      1.056 &      0.483 &      0.342 &      0.156 &      1.853 &      0.603 &      0.600 &      0.195 &      1.407 &      0.510 &      0.455 &      0.165 \\

%           &            &     $\sigma^2_w$ &     -1.882 &      0.965 &     -0.124 &      0.064 &     -0.762 &      0.976 &     -0.050 &      0.064 &      0.312 &      0.963 &      0.021 &      0.063 &      0.356 &      1.004 &      0.023 &      0.066 &      2.572 &      0.982 &      0.170 &      0.065 &      1.352 &      0.963 &      0.089 &      0.063 \\
\cline{2-27}

           &       case iii &         $\mu_w$ &      0.108 &      0.879 &      0.013 &      0.102 &      0.609 &      0.957 &      0.071 &      0.111 &      2.721 &      1.073 &      0.316 &      0.125 &      2.478 &      1.117 &      0.288 &      0.130 &      9.812 &      1.631 &      1.141 &      0.190 &      2.846 &      1.076 &      0.331 &      0.125 \\

           &            &         $\xi_w$ &      0.090 &      0.345 &      0.029 &      0.112 &      0.076 &      0.355 &      0.025 &      0.115 &      0.700 &      0.425 &      0.227 &      0.138 &      0.684 &      0.435 &      0.222 &      0.141 &      2.949 &      0.776 &      0.955 &      0.251 &      0.891 &      0.442 &      0.288 &      0.143 \\

 %          &            &     $\sigma^2_w$ &     -2.174 &      1.030 &     -0.143 &      0.068 &     -1.518 &      1.057 &     -0.100 &      0.070 &      0.449 &      0.962 &      0.030 &      0.063 &     -0.003 &      0.997 &      0.000 &      0.066 &      4.722 &      0.970 &      0.311 &      0.064 &      0.575 &      0.991 &      0.038 &      0.065 \\
\hline

         6.0 - 230 &         case i &         $\mu_w$ &     -0.022 &      1.042 &     -0.004 &      0.174 &      0.748 &      1.238 &      0.125 &      0.206 &      2.201 &      1.381 &      0.367 &      0.230 &      3.135 &      1.534 &      0.522 &      0.256 &      6.001 &      1.941 &      1.000 &      0.324 &      3.226 &      1.564 &      0.538 &      0.261 \\

           &            &         $\xi_w$ &      0.033 &      0.231 &      0.029 &      0.203 &      0.069 &      0.260 &      0.061 &      0.228 &      0.250 &      0.290 &      0.219 &      0.254 &      0.516 &      0.342 &      0.452 &      0.300 &      0.851 &      0.470 &      0.746 &      0.413 &      0.631 &      0.383 &      0.553 &      0.336 \\

  %         &            &     $\sigma^2_w$ &     -4.020 &      1.403 &     -0.265 &      0.093 &     -3.022 &      1.449 &     -0.199 &      0.096 &     -1.564 &      1.391 &     -0.103 &      0.092 &     -1.008 &      1.427 &     -0.066 &      0.094 &      1.206 &      1.434 &      0.080 &      0.095 &     -0.602 &      1.499 &     -0.040 &      0.099 \\
\cline{2-27}

           &        case ii &         $\mu_w$ &      0.246 &      0.751 &      0.041 &      0.125 &      1.241 &      0.933 &      0.207 &      0.156 &      2.264 &      0.983 &      0.377 &      0.164 &      2.495 &      0.985 &      0.416 &      0.164 &      5.252 &      1.302 &      0.875 &      0.217 &      3.506 &      1.077 &      0.584 &      0.180 \\

           &            &         $\xi_w$ &      0.093 &      0.170 &      0.082 &      0.149 &      0.233 &      0.208 &      0.204 &      0.182 &      0.368 &      0.228 &      0.323 &      0.200 &      0.531 &      0.249 &      0.466 &      0.218 &      0.968 &      0.349 &      0.849 &      0.306 &      0.730 &      0.280 &      0.640 &      0.246 \\

   %        &            &     $\sigma^2_w$ &     -3.600 &      1.000 &     -0.237 &      0.066 &     -2.452 &      1.085 &     -0.162 &      0.072 &     -1.438 &      0.994 &     -0.095 &      0.066 &     -1.371 &      1.003 &     -0.090 &      0.066 &      0.907 &      1.050 &      0.060 &      0.069 &     -0.321 &      0.995 &     -0.021 &      0.066 \\
\cline{2-27}

           &       case iii &         $\mu_w$ &      0.029 &      0.739 &      0.005 &      0.123 &      0.549 &      0.812 &      0.091 &      0.135 &      2.320 &      1.029 &      0.387 &      0.172 &      2.153 &      0.981 &      0.359 &      0.163 &      9.314 &      1.582 &      1.552 &      0.264 &      2.519 &      1.005 &      0.420 &      0.168 \\

           &            &         $\xi_w$ &      0.050 &      0.163 &      0.044 &      0.143 &      0.050 &      0.169 &      0.043 &      0.149 &      0.341 &      0.218 &      0.299 &      0.192 &      0.335 &      0.216 &      0.294 &      0.190 &      1.581 &      0.451 &      1.387 &      0.396 &      0.448 &      0.236 &      0.393 &      0.207 \\

    %       &            &     $\sigma^2_w$ &     -3.864 &      1.012 &     -0.255 &      0.067 &     -3.188 &      1.023 &     -0.210 &      0.067 &     -1.404 &      1.060 &     -0.093 &      0.070 &     -1.757 &      0.986 &     -0.116 &      0.065 &      3.157 &      1.018 &      0.208 &      0.067 &     -1.140 &      1.044 &     -0.075 &      0.069 \\
\hline

         8.6 - 100 &         case i &         $\mu_w$ &      0.167 &      0.885 &      0.019 &      0.103 &      0.722 &      0.979 &      0.084 &      0.114 &      1.774 &      1.108 &      0.206 &      0.129 &      2.632 &      1.195 &      0.306 &      0.139 &      4.931 &      1.486 &      0.573 &      0.173 &      2.722 &      1.179 &      0.316 &      0.137 \\

           &            &         $\xi_w$ &      0.083 &      0.568 &      0.016 &      0.109 &      0.159 &      0.613 &      0.031 &      0.117 &      0.540 &      0.659 &      0.104 &      0.126 &      1.126 &      0.703 &      0.216 &      0.135 &      1.846 &      0.850 &      0.354 &      0.163 &      1.354 &      0.745 &      0.259 &      0.143 \\

     %      &            &     $\sigma^2_w$ &      0.098 &      1.105 &      0.010 &      0.111 &      1.135 &      1.174 &      0.113 &      0.117 &      2.534 &      1.246 &      0.253 &      0.125 &      3.211 &      1.304 &      0.321 &      0.130 &      5.799 &      1.355 &      0.580 &      0.136 &      3.089 &      1.315 &      0.309 &      0.131 \\
\cline{2-27}

           &        case ii &         $\mu_w$ &      0.385 &      0.628 &      0.045 &      0.073 &      1.155 &      0.673 &      0.134 &      0.078 &      1.989 &      0.751 &      0.231 &      0.087 &      2.205 &      0.783 &      0.256 &      0.091 &      4.258 &      0.969 &      0.495 &      0.113 &      2.932 &      0.800 &      0.341 &      0.093 \\

           &            &         $\xi_w$ &      0.203 &      0.402 &      0.039 &      0.077 &      0.529 &      0.441 &      0.101 &      0.085 &      0.866 &      0.476 &      0.166 &      0.091 &      1.183 &      0.494 &      0.227 &      0.095 &      2.026 &      0.609 &      0.388 &      0.117 &      1.556 &      0.522 &      0.298 &      0.100 \\

      %     &            &     $\sigma^2_w$ &      0.353 &      0.771 &      0.035 &      0.077 &      1.424 &      0.789 &      0.142 &      0.079 &      2.547 &      0.850 &      0.255 &      0.085 &      2.472 &      0.865 &      0.247 &      0.087 &      4.801 &      0.948 &      0.480 &      0.095 &      3.255 &      0.877 &      0.325 &      0.088 \\
\cline{2-27}

           &       case iii &         $\mu_w$ &      0.194 &      0.623 &      0.023 &      0.072 &      0.544 &      0.681 &      0.063 &      0.079 &      1.975 &      0.779 &      0.230 &      0.091 &      1.915 &      0.823 &      0.223 &      0.096 &      7.407 &      1.173 &      0.861 &      0.136 &      2.159 &      0.765 &      0.251 &      0.089 \\

           &            &         $\xi_w$ &      0.098 &      0.405 &      0.019 &      0.078 &      0.102 &      0.408 &      0.020 &      0.078 &      0.756 &      0.484 &      0.145 &      0.093 &      0.800 &      0.491 &      0.153 &      0.094 &      3.118 &      0.713 &      0.598 &      0.137 &      0.990 &      0.483 &      0.190 &      0.093 \\

       %    &            &     $\sigma^2_w$ &      0.156 &      0.792 &      0.016 &      0.079 &      0.909 &      0.868 &      0.091 &      0.087 &      2.633 &      0.894 &      0.263 &      0.089 &      2.466 &      0.910 &      0.247 &      0.091 &      7.641 &      0.890 &      0.764 &      0.089 &      2.643 &      0.895 &      0.264 &      0.090 \\
\hline
\end{tabular}
%\end{adjustbox}
%\end{footnotesize}
}
\end{table}

\begin{figure}
	\subfigure[$\theta=0.1$]{\includegraphics[height=5.3cm,width=5.3cm]{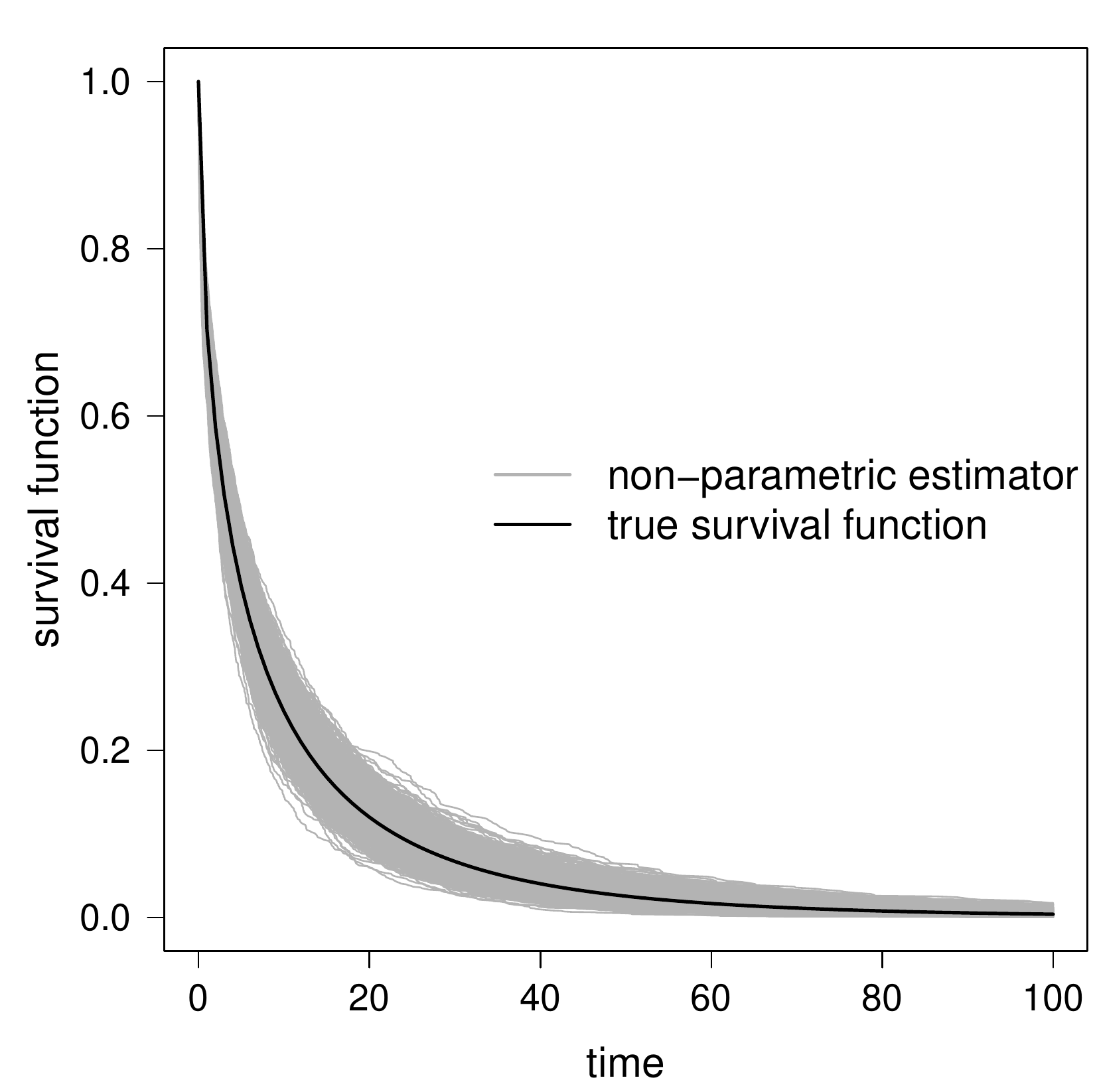}}
	\subfigure[$\theta=0.25$]{\includegraphics[height=5.3cm,width=5.3cm]{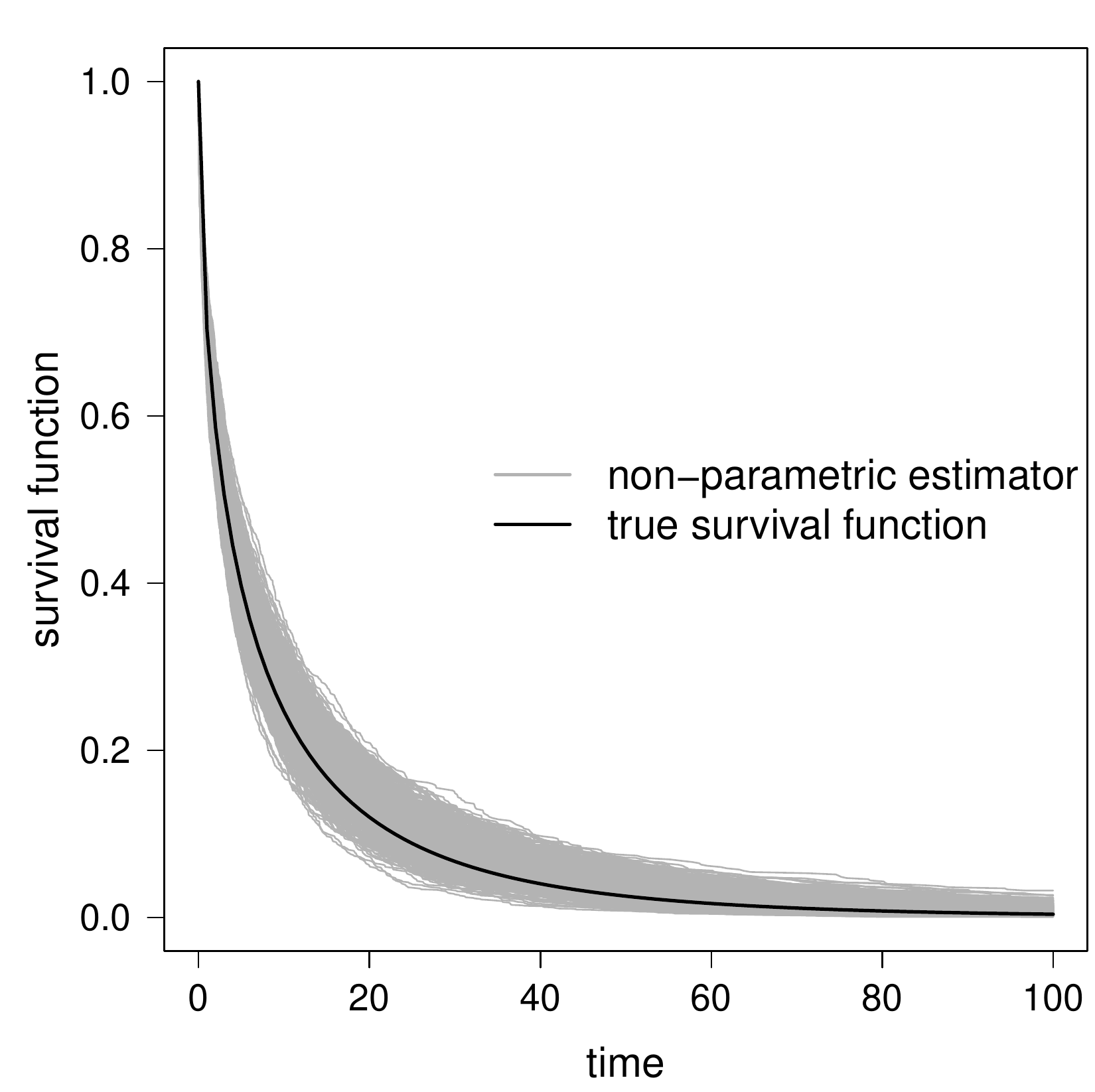}}
	\subfigure[$\theta=0.5$]{\includegraphics[height=5.3cm,width=5.3cm]{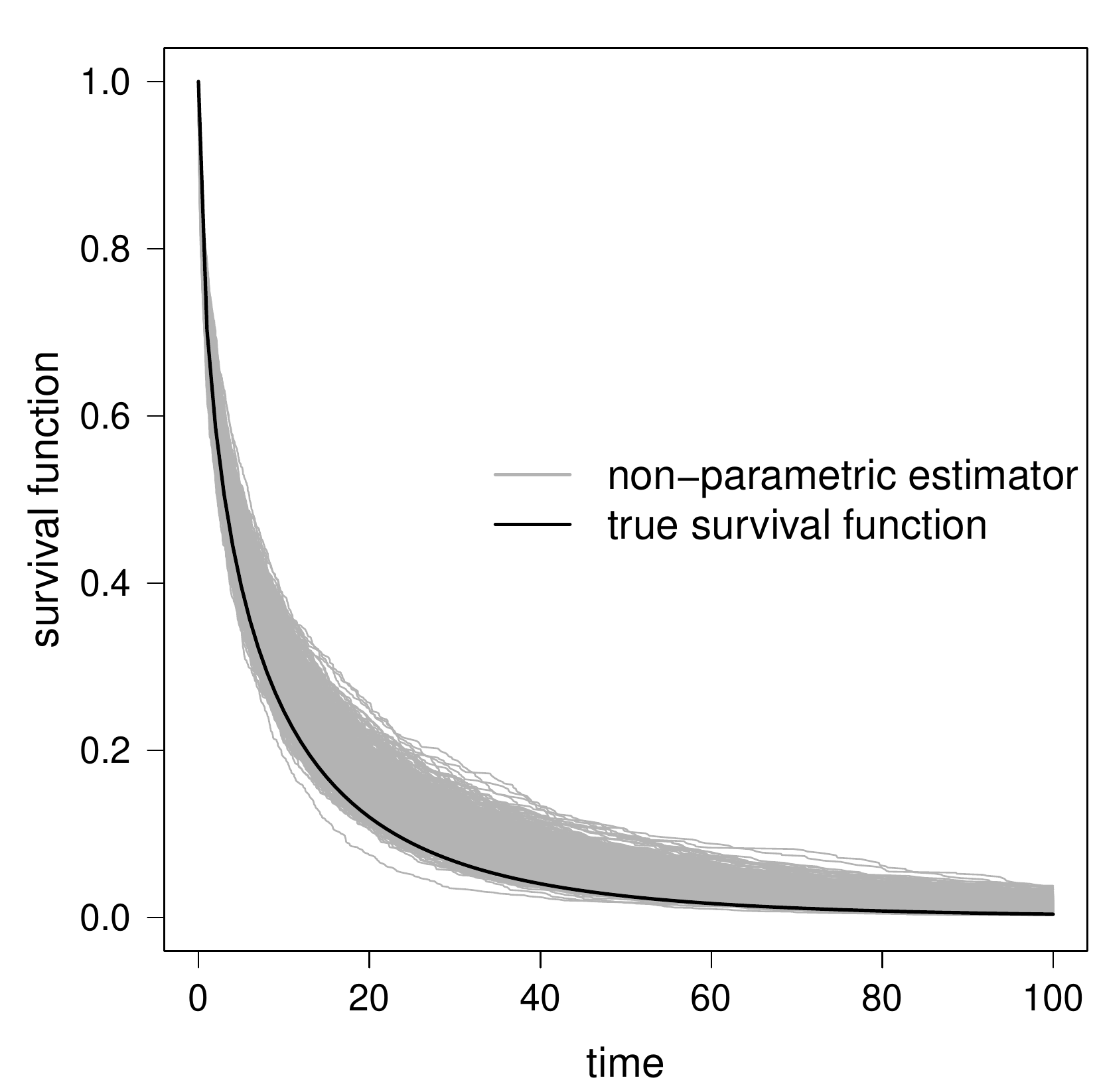}}
	\vspace{-0.25cm}
	\caption{Non-parametric estimator for the baseline survival function for 1,000 replicates when the data generation and estimation is based on the WL$(\theta$) model with a) $\theta=0.1$, b) $\theta=0.25$ and c) $\theta=0.5$, $\mu_w=8.6$, $\sigma_w^2=230$ and case i) for the specification of the number of clusters/observations in each cluster.}
	\label{np.versus.true}
\end{figure}

\newpage

\section{Real-world data analysis}\label{aplic}
\noindent

In this section, we present an application of the proposed model to real data for illustrative purposes.
The data set is related to rehospitalization times after surgery in patients diagnosed with colorectal cancer. It was presented for the first time in \cite{Gonzalez2005} and it is available in the \texttt{frailtySurv} \citep{frailtySurv2018} package of R \citep{R2022}. As the times related to a readmission are from the same individual, it is natural to use a frailty model in this context. We consider interocurrence or censoring time (in days) as the response variable. In total, the data contains 861 measures related to 403 patients (mean: 480.01, median: 216, standard deviation: 558.17, 47\% of times were censored). The distribution of the observations in each cluster is given in Table \ref{clusters}. Note that an individual (a cluster) has 23 observations.
\begin{table}[!h]
\caption{Distribution for the number of observations in each cluster.}\label{clusters}
    \centering
\resizebox{\linewidth}{!}{
\begin{tabular}{cccccccccccccc}
\hline
Observations in the cluster & 1 & 2  & 3 &  4 &  5 &  6 &  7 &  9 & 10 & 11 & 12 & 17 & 23 \\
\hline
Number of clusters & 199 & 105 & 45 & 21 & 15 &  8 &  4 &  1 &  1 &  1 &  1 &  1 &  1 \\
\hline
\end{tabular}
}
\end{table}

The considered covariates were: \texttt{dukes}, Dukes' tumoral stage (A-B: 324, C: 331 and D: 206); \texttt{charlson}, the comorbidity Charlson's index (0: 577 and 1-2-3: 284); \texttt{sex} (male: 549 and female: 312); and \texttt{chemo}, if the patient received chemotherapy (non-treated: 468 and treated: 393). Figure \ref{KM.app} shows the Kaplan-Meier estimator for the four covariates. Note that apparently the four covariates influence the rehospitalization time.

\begin{figure}[!h]
\begin{minipage}[b]{0.23\linewidth}
\centering
\includegraphics[width=3.7cm]{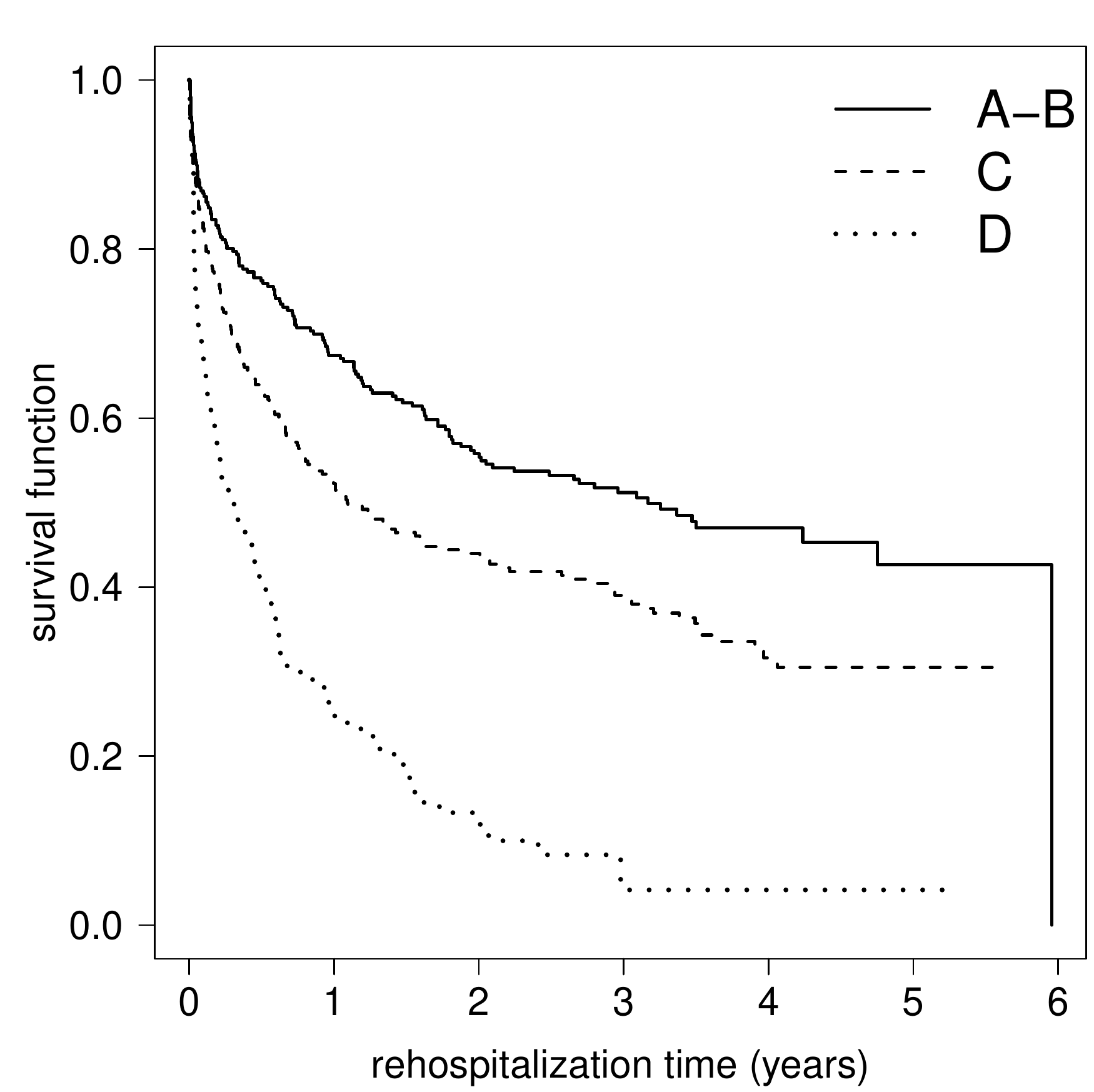}
\caption*{(a) \texttt{dukes}}
\end{minipage} %\hfill
\hspace{0.2cm}
\begin{minipage}[b]{0.23\linewidth}
\centering
\includegraphics[width=3.7cm]{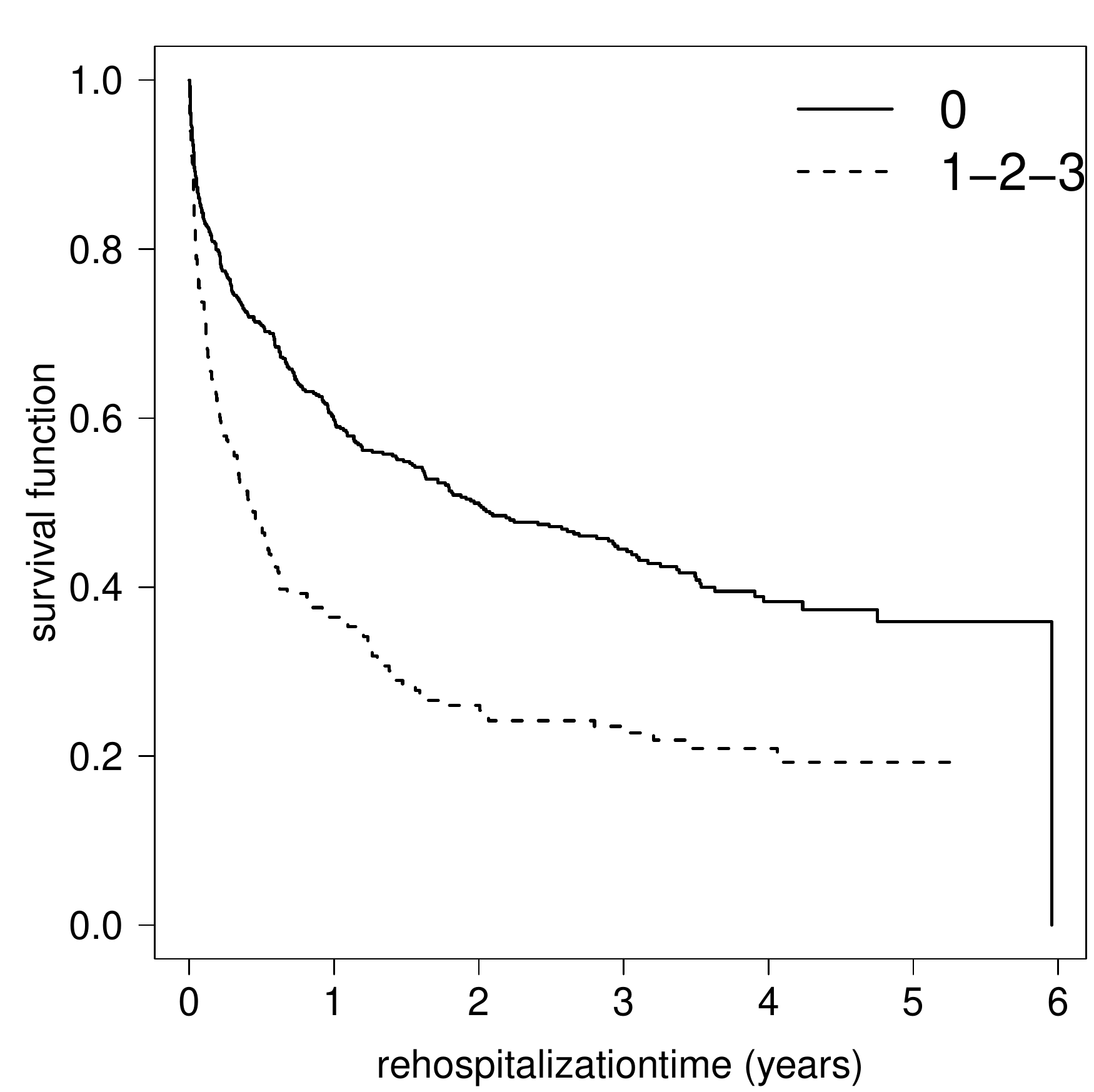}
\caption*{(b) \texttt{charlson}}
\end{minipage}
\hspace{0.2cm}
\begin{minipage}[b]{0.23\linewidth}
\centering
\includegraphics[width=3.7cm]{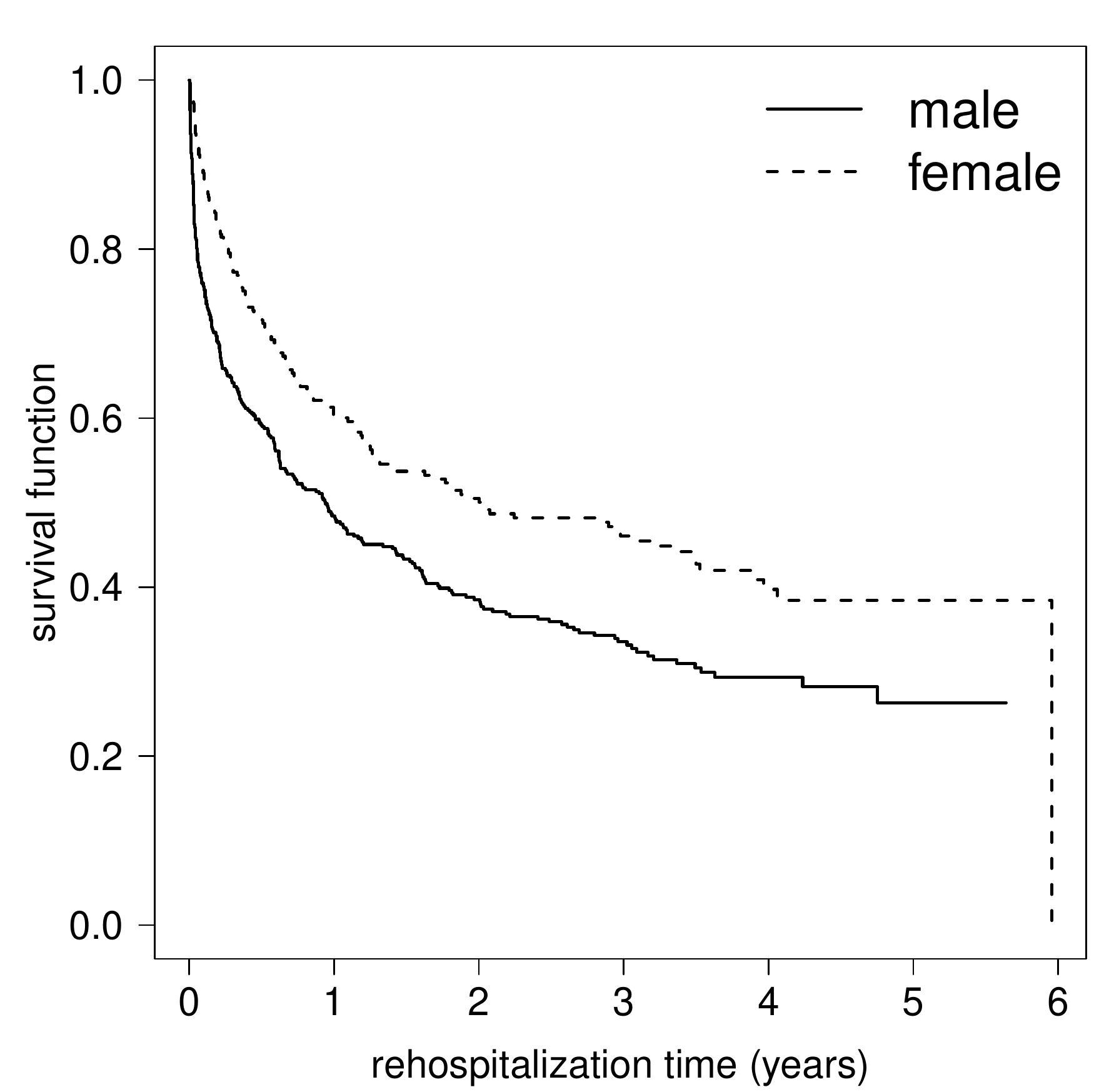}
\caption*{(c) \texttt{sex}}
\end{minipage} %\hfill
\hspace{0.2cm}
\begin{minipage}[b]{0.23\linewidth}
\centering
\includegraphics[width=3.7cm]{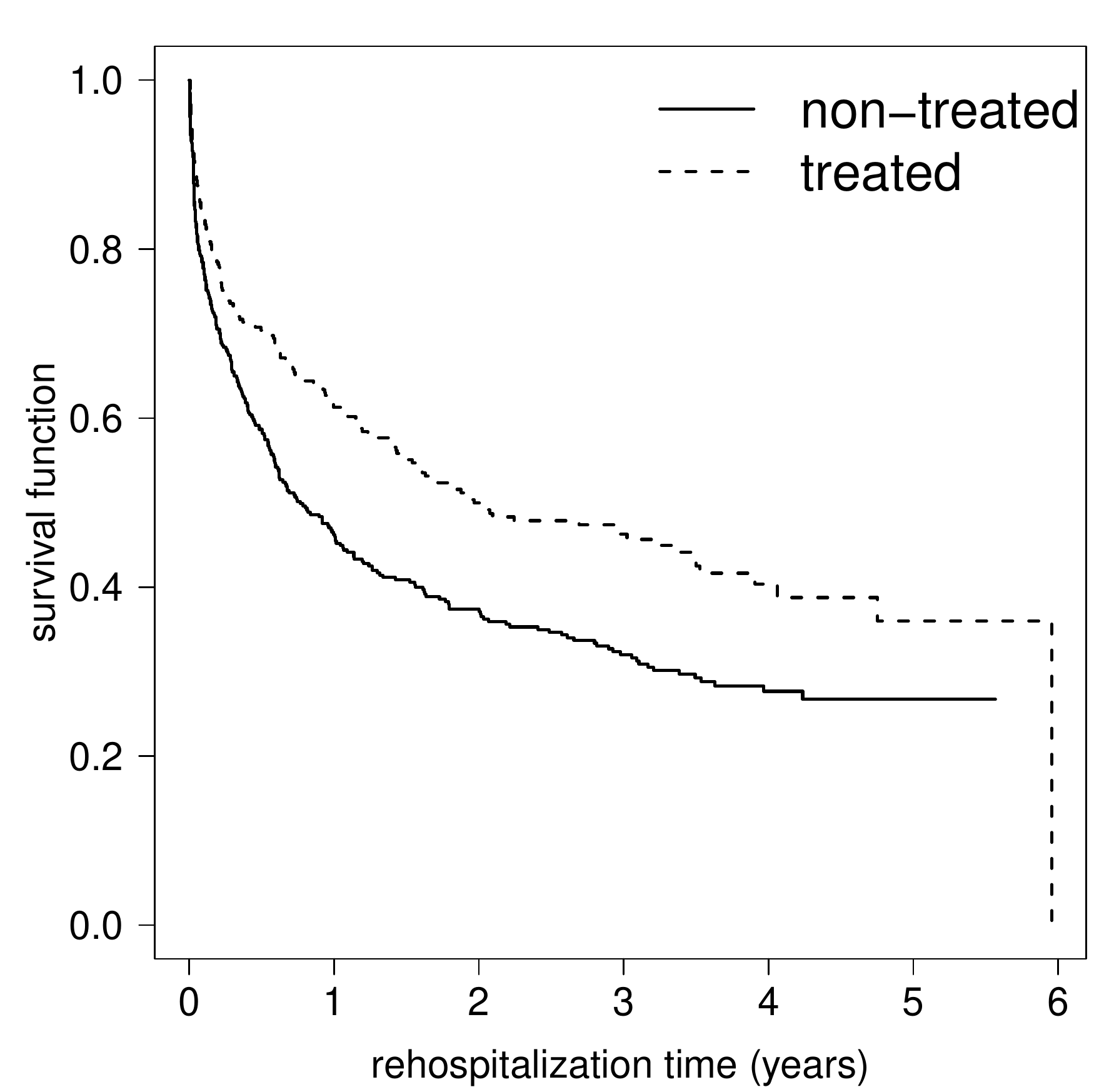}
\caption*{(d) \texttt{chemo}}
\end{minipage} %\hfill
\caption{Kaplan-Meier estimator for the four covariates in readmission data set.}
\label{KM.app}
\end{figure}

For the analysis, we considered the semi-parametric model with frailty WL, gamma and IG distributions (we refer to such models as semi-WL, semi-gamma and semi-IG, respectively) and the Weibull distribution with frailty WL, gamma and IG distributions (we refer to such models as WEI-WL, WEI-gamma and WEI-IG, respectively). To obtain the estimates for the semi-gamma and semi-IG models, we use the \texttt{coxph} function included in the \texttt{survival} package \citep{survival-package}, where the standard error term is not presented for $\theta$ (the frailty variance) and for the WEI-gamma and WEI-IG we use the \texttt{parfm} function from the \texttt{parfm} package \citep{parfm-package}. Table \ref{App.readmission} shows the estimates for such models. We highlight that all the models suggest that the inclusion of the frailty terms is necessary. We also highlight that for the semi-WL and WEI-WL all the estimates for the regression coefficients are significative using a 5\%, differently from the other models where the coefficients related to \texttt{dukesC} and \texttt{chemo} were non-significant. On the other hand, the estimated variance for the semi-WL and WEI-WL are very close (0.619 and 0.615, respectively), different to the semi-gamma with WEI-gamma and semi-IG and WEI-IG models, where there is a difference around 15\% among the estimated variance using the semi-parametric and the Weibull model. In this line, the WEI-gamma and WEI-IG estimate a greater Kendall's $\tau$ in comparison with the semi-gamma and semi-IG models, respectively. In this sense, the use of the WL distribution for frailty provides robustness. For this reason,
henceforth we consider the analysis with the semi-WL model.

\begin{table}[!h]
\caption{Estimation for readmission data set using the semiparametric model and the Weibull model with different frailty distributions: WL, gamma and IG.} \label{App.readmission}
    \centering
\resizebox{\linewidth}{!}{
% Table generated by Excel2LaTeX from sheet 'Hoja1'
\begin{tabular}{crrrrrrrrrrrrrrrrrr}
\hline

           &         \multicolumn{ 3}{c}{semi-WL} &      \multicolumn{ 3}{c}{semi-gamma} &    \multicolumn{ 3}{c}{semi-IG}     &          \multicolumn{ 3}{c}{WEI-WL} &       \multicolumn{ 3}{c}{WEI-gamma} &          \multicolumn{ 3}{c}{WEI-IG} \\

 Parameter &  Estimated &            &       s.e. &  Estimated &            &       s.e. &  Estimated &            &       s.e. &  Estimated &            &       s.e. &  Estimated &            &       s.e. &  Estimated &            &       s.e. \\
\hline

    \texttt{dukesC} &      0.308 &          * &      0.121 &      0.293 &            &      0.156 &      0.294 &            &      0.159 &      0.313 &          * &      0.190 &      0.293 &            &      0.161 &      0.297 &            &      0.165 \\

    \texttt{dukesD} &      1.125 &          * &      0.170 &      1.016 &          * &      0.187 &      1.067 &          * &      0.190 &      1.110 &          * &      0.213 &      1.076 &          * &      0.193 &      1.142 &          * &      0.198 \\

\texttt{charlson1-2-3} &      0.420 &          * &      0.119 &      0.402 &          * &      0.124 &      0.358 &          * &      0.123 &      0.372 &          * &      0.121 &      0.430 &          * &      0.127 &      0.379 &          * &      0.126 \\

    \texttt{sex} &     $-$0.578 &          * &      0.127 &     $-$0.516 &          * &      0.135 &     $-$0.495 &          * &      0.137 &     $-$0.495 &          * &      0.145 &     $-$0.525 &          * &      0.139 &     $-$0.502 &          * &      0.142 \\

   \texttt{chemo} &     $-$0.267 &          * &      0.105 &     $-$0.203 &            &      0.138 &     $-$0.202 &            &      0.141 &     $-$0.284 &          * &      0.169 &     $-$0.189 &            &      0.143 &     $-$0.188 &            &      0.147 \\
\hline

     $\theta$ &      0.619 &            &      0.123 &      0.589 &            &            &      0.654 &            &            &      0.615 &            &      0.179 &      0.688 &            &      0.142 &      0.786 &            &      0.197 \\

       Kendall's $\tau$ &      0.246 &            &            &      0.228 &            &            &      0.177      &            &            &      0.245 &            &            &      0.256 &            &            &      0.197 &            &            \\
\hline
$\rho$ &            &            &            &            &            &            &            &            &            &      0.440 &            &      0.019 &      0.641 &            &      0.026 &      0.643 &            &      0.026 \\

    $\lambda$ &            &            &            &            &            &            &            &            &            &      0.514 &            &      0.112 &      0.449 &            &      0.069 &      0.446 &            &      0.070 \\
\hline
\multicolumn{19}{l}{{*Significative coefficients based on a significance level of 5\%}}\\

\end{tabular}
}
\end{table}

Figure \ref{frailty.semiWL} shows the estimated frailties for each patient. Note that patients 274 and 318 appear as the patients with a higher risk, whereas patients 80 and 268 appear as the patients with a lower risk. This is corroborated by the descriptive analysis in Table \ref{App.readmission}.
Finally, Figure \ref{survival.app} presents the univariate survival function for one time related to patients 80 and 274 and the marginal survival function for one specific profile.

\begin{figure}[!htbp]
\centering
\includegraphics[width=5.5cm]{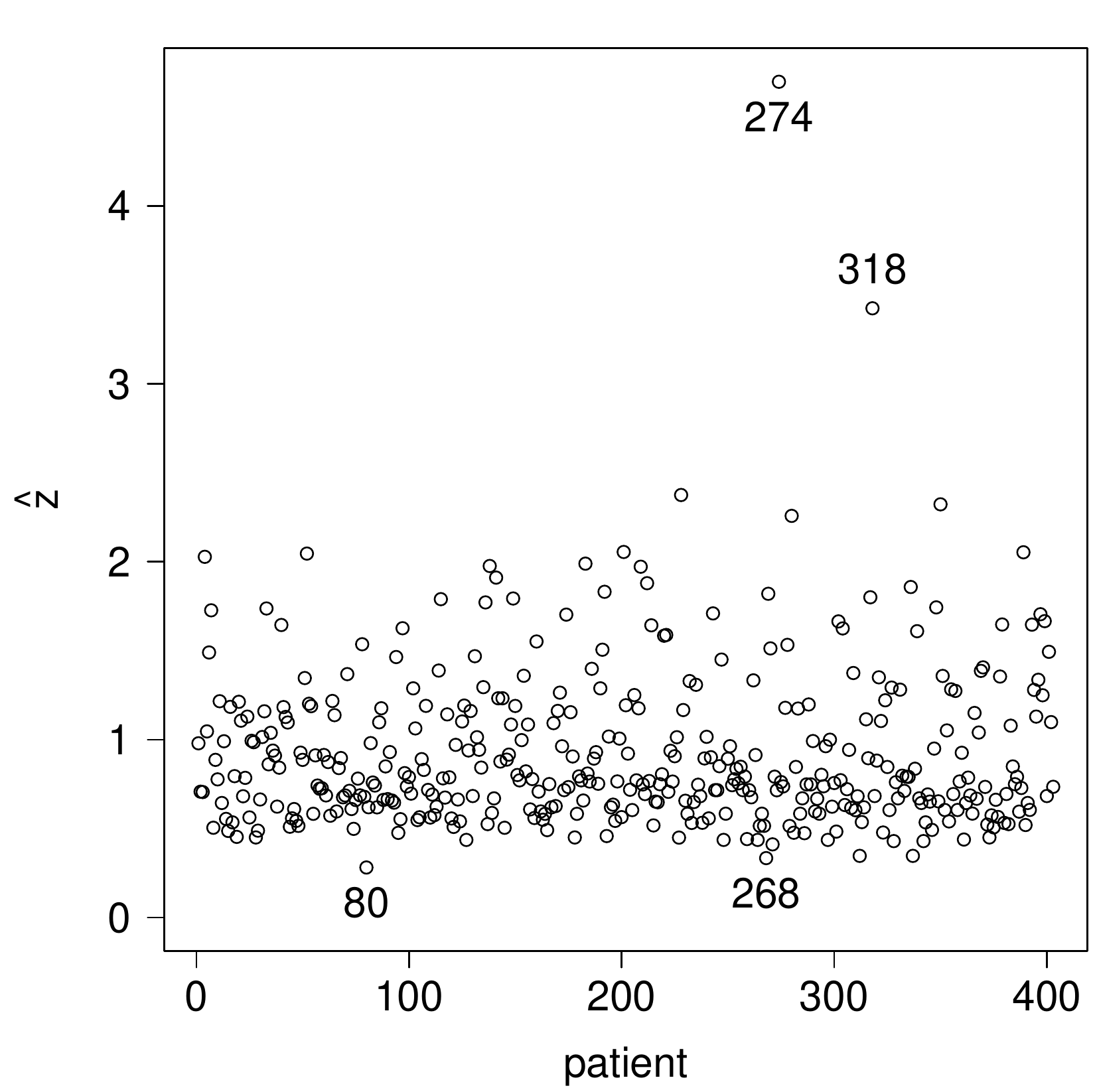}
\caption{Estimated frailty terms for each cluster (patient) in readmission data set using the semi-WL model.}
\label{frailty.semiWL}
\end{figure}

\begin{table}
\caption{Descriptive analysis for rehospitalization times for selected patients. $Q_j$ denotes the $j$-th quartile.} \label{App.readmission}
\begin{center}
\begin{tabular}{ccccl}
\hline
           &           \multicolumn{ 3}{c}{Rehospitalization times} &            \\

   Patient &         $Q_1$ &         $Q_2$ &         $Q_3$ &        $n_i$ \\
\hline

        80 &      5.276 &      5.276 &      5.276 &          1 \\

       268 &      5.073 &      5.073 &      5.073 &          1 \\

       274 &   0.006 &   0.005 &   0.008 &         17 \\

       318 &   0.015 &   0.019 &   0.041 &          7 \\

       all & 0.613 & 1.565 & 3.153 &       2* \\
\hline
\multicolumn{5}{l}{*\scriptsize{represents the median of the measures for all the patients.}}\\
\hline
\end{tabular}
\end{center}
\end{table}

\begin{figure}[!h]
\begin{minipage}[b]{0.46\linewidth}
\centering
\includegraphics[width=7cm]{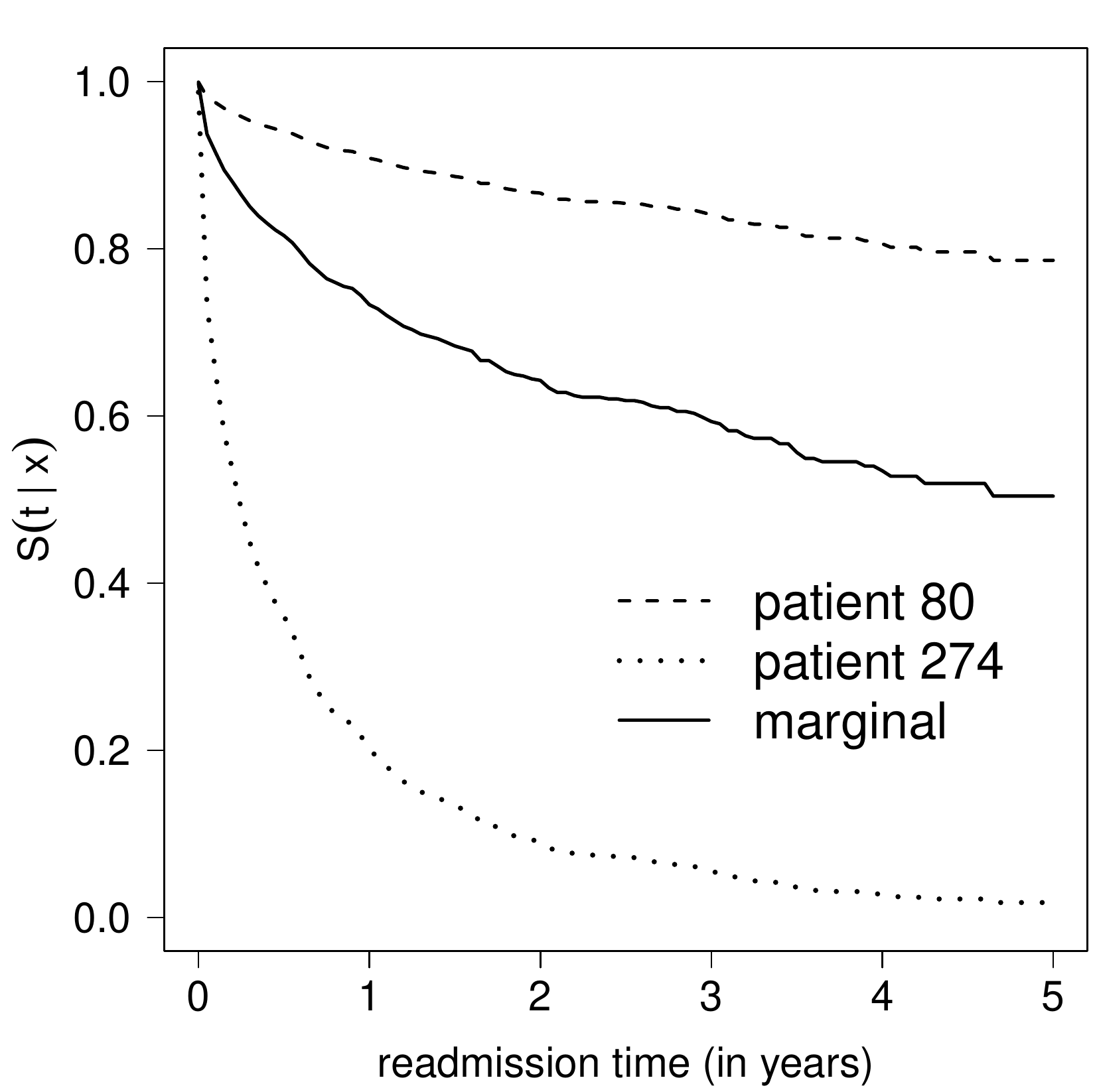}
\caption*{(a)}
\end{minipage} %\hfill
\hspace{0.3cm}
\begin{minipage}[b]{0.46\linewidth}
\centering
\includegraphics[width=7cm]{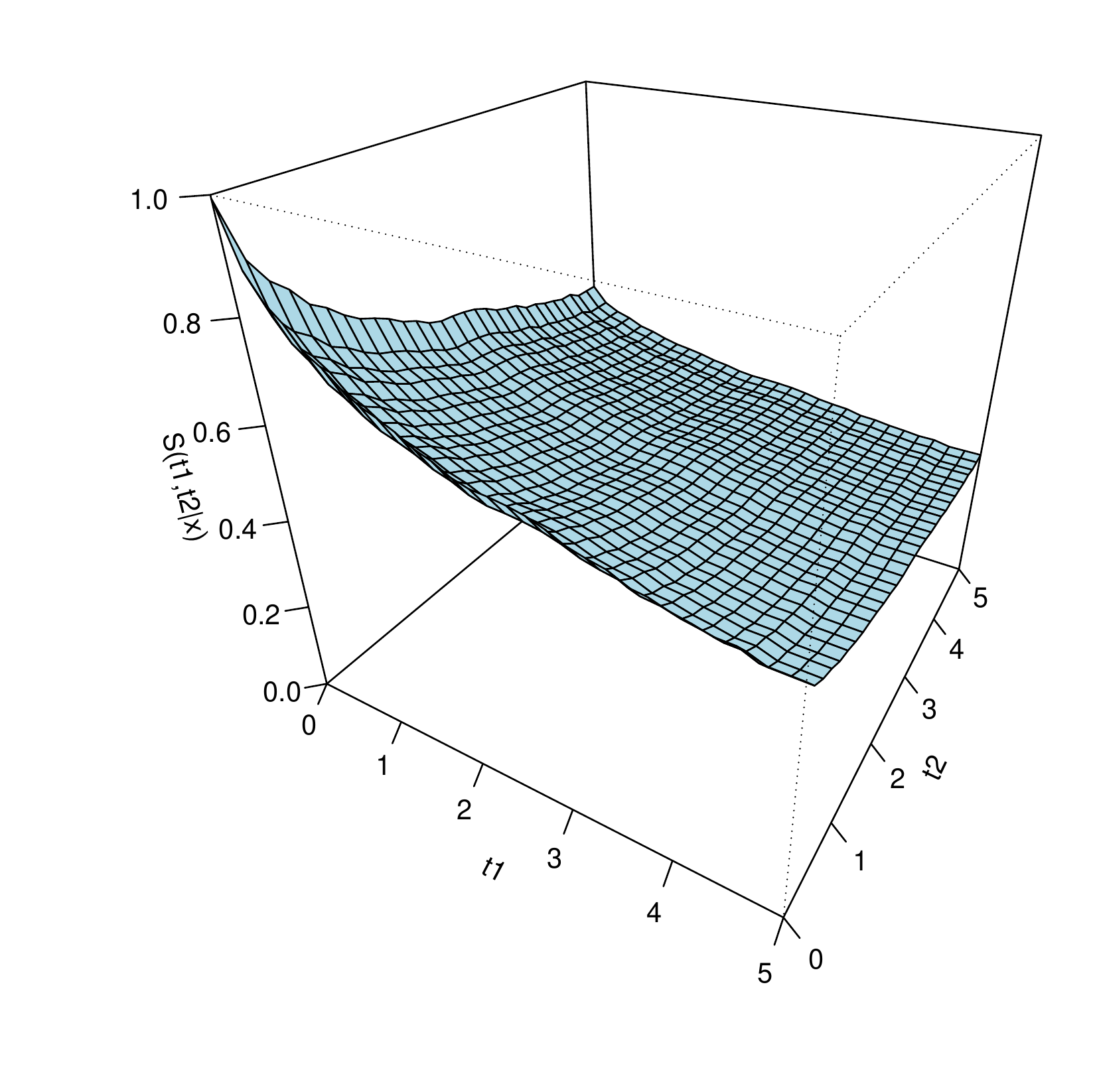}
\caption*{(b)}
\end{minipage}
\caption{(a) estimated conditional survival function for one readmission time based on the frailty WL semi-parametric model for patients 80 and 274 and the unconditional survival function for the profile \texttt{duke}=C, \texttt{charlon}=0, \texttt{sex}=female and \texttt{chemo}=treated. (b) Unconditional bivariate survival function based on the semi-WL and the same profile.}
\label{survival.app}
\end{figure}

\newpage

\section{Concluding remarks}
\noindent

In this paper, we have proposed a novel parametric (the Weibull hazard function was selected as the baseline)
and semi-parametric frailty model with WL frailty
distribution for modeling unobserved heterogeneity in cluster failure time data, which enjoy mathematical tractability
(has a simple Laplace transform) like the gamma frailty model. In particular, the WL distribution with unit mean and variance
$\sigma^2$ is used as the frailty distribution.
The semi-parametric choice of baseline hazard function provides a robust and flexible way to model the data.
We get a closed form expression for the derivatives of the Laplace transform for the WL distribution. Furthermore,
the conditional distributions of frailties among the survivors and the frailty of individuals dying
at time $t$ were determined explicitly.
A Monte Carlo simulation study has shown that the estimates based on the EM-algorithm of
the model parameters tend to their true values for both parametric and semiparametric cases.
In addition, to estimate the frailty variance it is better to perform more intra clusters than inter clusters observations.

Finally, we fitted the proposed regression model to a real dataset on rehospitalization
times after surgery in patients diagnosed with colorectal cancer to show the potential of using the new methodology. This
application also demonstrates the practical relevance of the new frailty model.
From the illustrative example analyzed, the WL frailty model is seen to be quite
robust in estimating the covariate effects as well as the frailty variance.
Mathematical tractability, flexibility, properties simplicity and computationally attractive of the
WL frailty model make the proposed model a competitive one among many models that already exist.
In this context, we see that the WL frailty model can also be useful in applications.
As part of future research, we plan to explore other estimation methods for the model,
including, for instance, the Bayesian approach. Furthermore, the model can be extended to the case of time-varying frailty and the WL frailty model with cure fraction.

\section*{Appendix}
\noindent In this Appendix we detail the proofs related to Propositions \ref{prop1} and \ref{prop2}.

\subsection{Proof of Proposition \ref{prop1}}
\noindent The conditional density of
$Z|T > t$ is
\begin{eqnarray*}
f(z|T> t) &=& \frac{f(z)S(t|z)}{S(t)} \\\\
&=& \frac{\theta \, a_\theta^{-b_\theta-1}z^{b_\theta - 1}(1+z)\exp(-z/a_\theta)\exp(-z\,\Lambda_0(t))}{2\,\Gamma(b_\theta)[1+a_\theta\,H_0(t)]^{-b_\theta-1}(1+\theta\,\Lambda_0(t)/2)}\\\\
&=&\frac{\theta \, A_\theta^{-b_\theta-1}z^{b_\theta - 1}(1+z)\exp(-z/A_\theta)}{2\,\Gamma(b_\theta)(1+\theta\,\Lambda_0(t)/2)} \\\\
&=& \frac{A_\theta^{-b_\theta-1}}{(A_\theta^{-1} + b_\theta)\Gamma(b_\theta)}z^{b_\theta - 1}(1+z)\exp(-z/A_\theta), \quad {z}>0,
\end{eqnarray*}
where $A_\theta = a_\theta/(1+a_\theta\,\Lambda_0(t))$ and $A_\theta^{-1} + b_\theta = (2 + \theta\,\Lambda_0(t))/\theta$. This provides that $Z|T > t \sim \textrm{WL}(b_\theta, A_\theta^{-1})$.

\subsection{Proof of Proposition \ref{prop2}}
\noindent

The conditional
density of $Z|T = t$ is given by
\begin{eqnarray*}
f(z|T = t) &=& \frac{f(z)f(t|z)}{f(t)} \\\\
&=& \frac{\theta \, a_\theta^{-b_\theta-1}z^{b_\theta - 1}(1+z)\exp(-z/a_\theta)z\,\lambda_0(t)\,\exp(-z\,\Lambda_0(t))}{2\,\Gamma(b_\theta)\lambda_0(t)[1+a_\theta\,\Lambda_0(t)]^{-b_\theta-2}(1+\theta\,\Lambda_0(t)/(\theta+2))}\\\\
&=&\frac{\theta \, a_\theta \,A_\theta^{-b_\theta-2}z^{b_\theta}(1+z)\exp(-z/A_\theta)}{2(\theta+2)^{-1}\,\Gamma(b_\theta)(\theta + 2+\theta\,\Lambda_0(t))} \\\\
&=&\frac{\theta \,A_\theta^{-b_\theta-2}z^{b_\theta}(1+z)\exp(-z/A_\theta)}{\Gamma(b_\theta + 1)(\theta + 2+\theta\,\Lambda_0(t))} \\\\
&=& \frac{A_\theta^{-b_\theta-2}}{(A_\theta^{-1} + b_\theta + 1)\Gamma(b_\theta + 1)}z^{b_\theta}(1+z)\exp(-z/A_\theta), \quad {z}>0,
\end{eqnarray*}
where $A_\theta = a_\theta/(1+a_\theta\,\Lambda_0(t))$, $a_\theta b_\theta = 2/(2 + \theta)$ and $A_\theta^{-1} + b_\theta + 1 = (\theta + 2 + \theta\,\Lambda_0(t))/\theta$. This provides that $Z|T = t \sim \textrm{WL}(b_\theta + 1, A_\theta^{-1})$.

\bibliographystyle{elsarticle-harv}
\bibliography{references}

%\subsection{Simulation study}

%\begin{table}
%\caption{Estimated bias, SE and RMSE for the multivariate WL frailty semiparametric model (censoring: %10\%)}
%\resizebox{\linewidth}{!}{
%}
%\end{table}

%\begin{table}
%\caption{Estimated bias, SE and RMSE for the multivariate WL frailty semiparametric model (censoring: %25\%)}
%\resizebox{\linewidth}{!}{
%}
%\end{table}

\end{document}